\pdfoutput=1
\documentclass[12pt,reqno,oneside]{amsart}
\usepackage[utf8]{inputenc}
\usepackage[english]{babel}
\usepackage[svgnames]{xcolor}
\usepackage[colorlinks=true,linkcolor=Green,citecolor=Orange,urlcolor=Blue]{hyperref}

\usepackage[final]{graphicx}
\usepackage{amsfonts,amsmath,amssymb,amsthm,mathtools}
\usepackage{bm}

\newtheorem{theorem}{Theorem}[section]
\newtheorem{lemma}[theorem]{Lemma}
\newtheorem{corollary}[theorem]{Corollary}

\newtheorem{remark}[theorem]{Remark}

\newcommand{\R}{{\mathord{\mathbb R}}}
\newcommand{\Sp}{{\mathord{\mathbb S}}}

\newcommand{\Z}{{\mathord{\mathbb Z}}}
\newcommand{\N}{{\mathord{\mathbb N}}}

\newcommand{\XXint}[3]{\setbox0=\hbox{$#1{#2#3}{\int}$ }
\vcenter{\hbox{$#2#3$ }}\kern-.6\wd0}

\usepackage{xcolor}

\usepackage[normalem]{ulem}

\begin{document}

\title{Entropy decay for the Kac evolution}
\author[F. Bonetto, A. Geisinger, M. Loss, and T. Ried]{}
\email{{bonetto@math.gatech.edu}}
\email{{alissa.geisinger@uni-tuebingen.de}}
\email{{loss@math.gatech.edu}}
\email{{tobias.ried@kit.edu}}
\thanks{\copyright~2017 by the authors. Reproduction of this article by any 
means permitted for non-commercial purposes} 
\maketitle

\centerline{\scshape Federico Bonetto}
\smallskip
{\footnotesize
 \centerline{School of Mathematics, Georgia Institute of Technology}
 \centerline{Atlanta, GA 30332, United States of America}
}

\medskip

\centerline{\scshape Alissa Geisinger}
\smallskip
{\footnotesize
 \centerline{Universit\"at T\"ubingen, Fachbereich Mathematik}
 \centerline{Auf der Morgenstelle 10, 72076 T\"ubingen, Germany}
}

\medskip

\centerline{\scshape Michael Loss}
\smallskip
{\footnotesize
 \centerline{School of Mathematics, Georgia Institute of Technology}
 \centerline{Atlanta, GA 30332, United States of America}
}

\medskip

\centerline{\scshape Tobias Ried}
\smallskip
{\footnotesize
 \centerline{Institute for Analysis, Karlsruhe Institute of Technology (KIT)}
 \centerline{Englerstra{\ss}e 2, 76131 Karlsruhe, Germany}
}

\begin{abstract} 
We consider solutions to the Kac master equation for initial conditions where 
$N$ particles are in a thermal equilibrium and $M\le N$ particles are out of 
equilibrium. We show that such solutions have exponential decay in entropy 
relative to the thermal state. More precisely, the decay is exponential in time 
with an explicit rate that is essentially independent on the particle number. 
This is in marked contrast to previous results which show that the entropy 
production for arbitrary initial conditions is inversely proportional to the 
particle number. The proof relies on Nelson's hypercontractive estimate and the 
geometric form of the Brascamp-Lieb inequalities due to Franck Barthe. Similar 
results hold for the Kac-Boltzmann equation with uniform scattering cross 
sections. 

\end{abstract}

\section{\bf Introduction}

Among the models describing a gas of interacting particles, the Kac master 
equation \cite{kac}, due to its simplicity, occupies a special place. It is 
useful in illuminating various issues in kinetic theory, e.g., providing a 
reasonably satisfactory derivation of the spatially homogeneous Boltzmann 
equation and giving a  mathematical framework for investigating the approach to 
equilibrium. These issues were, in fact, the motivation for Kac's original work 
\cite{kac}. Although it does not have a foundation in Hamiltonian mechanics, the 
Kac master equation is based on simple probabilistic principles and yields a 
linear evolution equation for the velocity distribution for $N$ particles 
undergoing collisions. It is in this context that Kac invented the notion of 
propagation of chaos and he used this notion to derive the spatially 
homogeneous, non-linear Kac-Boltzmann equation. The approach through master 
equations led Kac to formulate the notion of approach to equilibrium and 
suggested various avenues to investigate this problem as the number of 
particles, $N$, becomes large. He emphasized that this could be done in a 
quantitative way if one could show, e.g., that the gap of the generator is 
bounded below uniformly in $N$.  This, known as Kac's conjecture \cite{kac}, was 
proved by \'Elise Janvresse in \cite{Jeanvresse} and, as a further sign of the 
simplicity of the model,  the gap was computed explicitly in \cite{CCL1, CCL2}, 
see also \cite{Maslen}.  One of the problems in using the gap is that the 
approach to equilibrium is measured in terms of an $L^2$ distance. While this 
does seem to be a natural  way to look at this problem, the size of the  $L^2$ 
norm of approximately independent probability distributions increases 
exponentially with the size of the system. Thus, the half life of the $L^2$ norm 
is of order $N$. 

A natural measure is, of course, given by the entropy, which is extensive, i.e, 
proportional to $N$.  There has not been much success in proving exponential
decay of the entropy with good rates.  In \cite{villani} Cedric Villani showed that 
the entropy decays exponentially, albeit with a rate that is bounded below by a 
quantity that is inversely proportional to $N$. This estimate was complemented 
by Amit Einav \cite{amit}, who gave an example of a state that has entropy 
production essentially of order $1/N$. His example is the initial state in which 
most of the energy is concentrated in a few particles while most of the others 
have very little energy. One might surmise, based on physical intuition, that 
this state is physically very improbable and still has low entropy production 
because most of the particles are in some sort of equilibrium. This intuition 
can be made rigorous, see \cite{amit}, although by a quite difficult 
computation. One should add that low entropy production does not preclude 
exponential decay in entropy, i.e., large entropy production for the initial 
state might not be necessary for an exponential decay rate for the entropy.  

A breakthrough was achieved by Mischler and Mouhot  in \cite{MischlerMouhot1, 
MischlerMouhot2}. They undertook a general investigation of the Kac program for 
gases of hard spheres and true Maxwellian molecules in three dimensions. Among 
the results of Mischler and Mouhot  is a proof that these systems relax towards 
equilibrium in relative entropy as well as in Wasserstein distance with a rate 
that is independent of the particle number. As expected, they achieve this not 
for any initial condition, but rather for a natural class of chaotic states. The 
rate of relaxation is, however, polynomial in time. 

To summarize, there is so far no mathematical evidence that the entropy in the 
Kac model in general decays exponentially with a rate that is independent of $N$ 
and physical intuition suggests that for highly ``improbable'' states, such as 
the one used by Einav, this cannot be expected. One can restrict the class of 
initial conditions by considering chaotic states as done by Mischler and Mouhot, 
which shifts the problem of finding suitable initial conditions for proving 
exponential decay to the level of the non-linear Boltzmann equation. 

In this paper we take a different approach, one which is based on the idea of 
coupling a system of particles to a reservoir. Recall from \cite{BLV} the master 
equation of $M$ particles with velocities $\mathbf v = (v_1, v_2, \dots, v_M)$ 
interacting with a thermostat at temperature $1/\beta$,
\begin{align}\label{eq: gent}
 \frac{\partial f}{\partial t} = \mathcal L_Tf\ , \ f(\mathbf v, 0) = f_0(\mathbf v) \ .
\end{align}
The operator $\mathcal{L}_T$ is given by
\begin{align*}
\mathcal L_Tf =\mu \sum_{j=1}^{M}{(B_j-I)}f \ ,
\end{align*}
where
\begin{align*}
B_j[f](\mathbf v) :&= \int_\R \mathrm{d}w  \, \frac{1}{2\pi}  
\int_{-\pi}^{ \pi} \mathrm{d}\theta \, \sqrt{\frac \beta {2\pi}} \,
e^{-\beta  w_j^{*}(\theta)^2/2} f(\mathbf v_j(\theta,w)) \ , \\
\vphantom{\int} \mathbf v_j(\theta,w) &= (v_1,..., v_j \cos{(\theta)} + w 
\sin{(\theta)},...,v_M)  \ \text{ and }
 w_j^*(\theta) = -v_j \sin{(\theta)} + w \cos{(\theta)} \ .
\end{align*}
Thus, $B_j[f](\mathbf v)$ describes the effect of a collision between particle 
$j$ in the system and a particle in the reservoir. After the collision, the 
particle from the thermostat is discarded, which ensures that the thermostat 
stays in equilibrium. The interaction times with the thermostat are given by a 
Poisson process whose intensity $\mu$ is chosen so that the average time between 
two successive interactions of a given particle with the thermostat is 
independent of the number of particles in the system. For the case where 
$\rho(\theta) = (2\pi)^{-1}$, the entropy decays exponentially fast. In fact, 
abbreviating $\sqrt{\beta/(2\pi)}e^{-\beta/2\mathbf{v}^2} 
=\Gamma_{\beta}(\mathbf{v})$, we know from \cite{BLV}, that 
\begin{equation*}
S(f(\cdot, t)) := \int_{\R^M} f(\mathbf v, t) \log\left(\frac{f(\mathbf v, 
t)}{\Gamma_\beta (\mathbf v)} \right) 
d \mathbf v \le e^{- \mu t/2} S(f_0) \ .
\end{equation*}
Thus, one might guess that if a ``small" {\it system} of $M$ particles out of 
equilibrium interacts with a  {\it reservoir}, that is a large system of $N \ge 
M$ particles in thermal equilibrium, then the entropy decays exponentially fast 
in time. This intuition is also supported by the results in \cite{BLTV}. There 
it was shown that if the thermostat is replaced by a large but finite reservoir 
{\it initially} in thermal equilibrium, this evolution is close to the evolution 
given by the thermostat. This results holds in various norms and, in particular, 
it is {\it uniform} in time. We would like to emphasize that the reservoir will 
not stay in thermal equilibrium as time progresses, nevertheless it will not 
veer far from it.

Since this is the model that we consider in this work, we will now describe it 
in detail. We consider probability distributions $F:\R^{M+N} \rightarrow \R_+$ 
and write $F(\mathbf v, \mathbf w)$ where $\mathbf v = (v_1, \dots, v_M)$ 
describes the particles in the small system, whereas $ \mathbf w = (w_{M+1}, 
\dots, w_{N+M})$ describes the particles in the large system. The Kac master 
equation is given by
\begin{equation} \label{eq: masterequation}
\frac{\partial F}{\partial t} = \mathcal L F \ ,  \ F(\mathbf v, \mathbf w, 0) = 
F_0(\mathbf v, \mathbf w) = f_0(\mathbf v) e^{-\pi |\mathbf w|^2} \ , 
\end{equation}
where 
 \begin{equation}\label{eq:model}
\mathcal{L} = \frac{\lambda_S}{M-1} \sum_{1 \leq i<j \leq M} \left( R_{ij} - I 
\right) + \frac{\lambda_R}{N-1} \sum_{M < i<j \leq N + M} \left( R_{ij} - I 
\right) + \frac{\mu}{N} \sum_{i=1}^{M} \sum_{j=M+1}^{M+N} \left( R_{ij} - I 
\right) \ , 
\end{equation}
and $R_{ij}$ is given as follows. For $1 \le i<j \le M$ we have
\begin{align*}
(R_{ij}F)(\mathbf v, \mathbf w) = \int_{-\pi}^\pi  \rho(\theta) \, \mathrm{d} 
\theta \, F(r_{ij}(\theta)^{-1}(\mathbf v, \mathbf w)) \ , 
\end{align*}
where
\begin{align}\label{eq:R}
r_{ij}(\theta)^{-1}(\mathbf v, \mathbf w) = (v_1, \dots, v_i \cos \theta - v_j 
\sin \theta, \dots, v_i \sin \theta  + v_j \cos\theta , \dots, v_M, \mathbf w) 
\ . 
\end{align}
The other $R_{ij}$s are defined analogously. We assume that the probability 
measure $\rho$ is smooth and satisfies  
\begin{align} \label{assumprho}
\int_{-\pi}^\pi \rho(\theta) \, \mathrm{d}\theta \, \sin \theta \cos \theta =0 
\ 
.
\end{align}
In particular, we do {\it not} require $\mathcal{L}$ to be self-adjoint on 
$L^2(\mathbb{R}^{N+M})$, a condition called {\it microscopic 
reversibility}.  The initial state of the 
reservoir is assumed to be a thermal equilibrium state and we have chosen
 units in which the inverse temperature $\beta =2\pi$.   
Note that $\lambda_S$ is the rate at which one particle from the system will 
scatter with any other particle in the system and similarly for $\lambda_R$. 
Likewise, $\mu$ is the rate at which a single particle of the system will 
scatter with any particle in the reservoir. The rate at which a particular 
particle from the reservoir will scatter with a particle in the system is given 
by $\mu M/N$. Hence, when $N$ is large compared to $M$ this process is 
suppressed and one expects that the reservoir does not move far from its 
equilibrium.  Indeed, it is shown in \cite{BLTV} that the solution of the master 
equation \eqref{eq:model} stays close to the solution of a thermostated system 
in the Gabetta-Toscani-Wennberg metric, 
\begin{align*}
d_{GTW}(F,G) := \sup_{k \not= 0} \frac{|\widehat F(k) - \widehat G(k)|}{|k|^2} \ ,
\end{align*}
see \cite{GTW}. Here, $\widehat F$ denotes the Fourier transform of $F$.  
More precisely, with the initial conditions \eqref{eq: gent} and \eqref{eq: 
masterequation}, it was shown that 
\begin{align*}
d_{GTW}( f(\mathbf v, t) e^{-\pi |\mathbf w|^2}, F(\mathbf v, \mathbf w, t)) \le 
C(f_0) \frac{M}{N} \ , 
\end{align*}
where $C(f_0)$ is a constant that depends on the initial condition but is of 
order one. The distance varies inversely as $N$, the size of the reservoir and, 
moreover, this estimate holds {\it uniformly} in time. For a detailed 
description of the results we refer the reader to \cite{BLTV}. From this result 
and the fact that the entropy of the system interacting with a thermostat decays 
exponentially in time, one might surmise that the entropy of the system 
interacting with a finite reservoir also decays exponentially fast in time. In 
fact we shall show this to be true if we consider the entropy relative to the 
thermal state. 

\section{\bf Results}

For the solution of the master equation \eqref{eq: masterequation} we use use interchangeably
the notation
\begin{align} \label{eq:expLt}
F(\mathbf v, \mathbf w, t) = (e^{\mathcal L t}F_0)(\mathbf v, \mathbf w).
\end{align}
This evolution preserves the 
energy and hence it suffices to consider it on $L^1(\Sp^{N+M}(\sqrt{N+M}))$ 
with the normalized surface measure. Likewise, it is easy to see that the 
evolution is ergodic on $\Sp^{N+M}(\sqrt{N+M})$ in the sense that $e^{\mathcal 
L t} F_0 \to 1$ as $t \to \infty$ and $1$ is the only normalized equilibrium 
state.

For our purposes it is convenient to consider the evolution in $L^1(\R^{M+N})$ 
with Lebesgue measure. Then $e^{\mathcal L t}F_0$ converges to the spherical 
average of $F_0$ taken over spheres in $\R^{M+N}$. In this space we choose the 
initial condition
\begin{equation} \label{initialcondition}
F_0( \mathbf v, \mathbf w) = f_0(\mathbf v) e^{-\pi |\mathbf w|^2} \ .
\end{equation}
Moreover, we introduce the function $f$,
\begin{equation}  \label{marginal}
f(\mathbf v, t) := \int_{\R^N} \left[e^{\mathcal L t}F_0\right](\mathbf v, 
\mathbf w) \, \mathrm{d} \mathbf w \   
\end{equation}
and we call
\begin{equation*}
S(f(\cdot, t)) := \int_{\R^M} f(\mathbf v, t) \log\left( \frac{f(\mathbf v, 
t)}{e^{-\pi |\mathbf v|^2}} \right) \, \mathrm{d} \mathbf v  \ , 
\end{equation*}
the entropy of $f$ {\it relative} to the thermal state $e^{-\pi |\mathbf v|^2}$.
Our main result is the following theorem.

\begin{theorem} \label{thm: main} Let $N\ge M$
and let $\rho$ be a probability distribution with an absolutely convergent 
Fourier series such that \eqref{assumprho} holds.
The entropy of $f$ {\it relative} of to the thermal state $e^{-\pi |\mathbf 
v|^2}$ then satisfies  
\begin{equation*}
S(f(\cdot, t)) \le \left[\frac{M}{N+M} + \frac{N}{N+M} e^{- t \mu_{\rho} 
(N+M)/N} \right]S(f_0)  \ , 
\end{equation*}
where
\begin{equation*} 
\mu_{\rho} = \mu \int_{-\pi}^\pi  \rho(\theta) \, \mathrm{d} \theta \, 
\sin^2(\theta) \ , 
\end{equation*}
and $f_0$ is as introduced in \eqref{initialcondition}.
\end{theorem}

\begin{remark} 
\item{1.} Note that the theorem deals with the entropy relative to the thermal 
state and not with respect to the equilibrium state. The entropy relative to the 
equilibrium state tends to zero as $t \to \infty$. We do not know how to adapt 
our proof to this situation nor do we have any evidence that it does indeed tend 
to zero at an exponential rate.  If this were the case, the rate would most 
likely depend on the initial condition.

\item{2.} \label{rem:1.2.2}The decay rate is universal in the sense that it 
only depends on $\mu$ and the distribution $\rho$. The intra-particle 
interactions in the system and in the reservoir do not seem to matter.

\item{3.} The statement of the theorem becomes particularly simple as $N\to 
\infty$. This corresponds to the thermostat problem treated in \cite{BLV} with 
the exact same decay rate. It is known that for the thermostat the decay rate is 
optimal, see \cite{TV}, and hence the decay rate here is optimal as well.

\item{4.} Although we assume that $\rho$ is smooth, our result also holds for  
the case where $\rho$ is a finite sum of Dirac measures. In particular 
Theorem~\ref{thm: main} also holds if $\rho$ is a delta measure that has 
its mass at the angles $\theta = \pm \pi/2$, that is, our result does not depend 
on ergodicity of the evolution. 
\end{remark}

As a consequence of Remark \ref{rem:1.2.2}(2),
one obtains a result for the standard Kac model. Recall that the 
generator of the standard Kac model is given by 
\begin{align*}
\mathcal L_{\rm cl} = \frac{2}{N+M-1} \sum_{1 \le i < j \le N+M} (R_{ij} -I) 
\ . 
\end{align*}
We may arbitrarily split the variables into two groups $(v_1, \dots, v_M)$ and 
$(w_{M+1},  \dots, w_{M+N})$. Splitting the generator accordingly,
\begin{multline*}
\mathcal L_{\mathrm{cl}} = \frac{2}{N+M-1} \sum_{1 \le i < j \le M} (R_{ij} -I) 
+  \frac{2}{N+M-1} \sum_{M+1 \le i < j \le N+M} (R_{ij} -I) \\ 
+  \frac{2}{N+M-1} \sum_{i=1}^M \sum_{j=M+1}^{N+M} (R_{ij} -I)\ ,
\end{multline*}
we see that the standard Kac model can be cast in the from \eqref{eq:model} 
by setting
\begin{align*}
\lambda_S = \frac{2(M-1)}{N+M-1} \ \text{ , } \ \lambda_R =  
\frac{2(N-1)}{N+M-1} \ \text{ and } \ 
\mu = \frac{2N}{N+M-1} \ .
\end{align*}
Hence, we obtain the following Corollary.
\begin{corollary}
Let $N \ge M$ and consider the time evolution defined by $\mathcal L_{\rm cl}$ 
with initial condition \eqref{initialcondition}.
Assume that the function $f_0$ in the initial condition has finite entropy.
The entropy of the function
\begin{equation*}
f(\mathbf v, t) := \int_{\R^N} \left[e^{\mathcal L_{\mathrm{cl}} t}F_0\right](\mathbf v, \mathbf w) \, \mathrm{d} \mathbf w \  
\end{equation*}
{\it relative} to the thermal state $e^{-\pi |\mathbf v|^2}$, satisfies
\begin{equation*}
S(f(\cdot, t))
\le \left[\frac{M}{N+M} + \frac{N}{N+M} e^{- t \mu_{\rho} 2(N+M)/(N+M-1)} \right]S(f_0)  \ ,
\end{equation*}
where 
\begin{equation*} 
\mu_{\rho} = \int_{-\pi}^\pi  \rho(\theta) \, \mathrm{d} \theta \, \sin^2(\theta) \ 
\end{equation*}
and $\rho$ is a probability distribution such that \eqref{assumprho} holds.
\end{corollary}

On a mathematical level, an efficient way of proving approach to equilibrium is through
a logarithmic Sobolev inequality, which presupposes that the generator of the time evolution is given by
a Dirichlet form. This kind of structure is notably absent in the Kac master equation. We shall see however,
that the logarithmic Sobolev inequality in the form of Nelson's hypercontractive estimate is an important
tool for the proof of Theorem \ref{thm: main}. We will use an iterated version of it, which expresses the result in
terms of marginals of the functions involved. This, coupled with an auxiliary computation and a sharp version of the Brascamp-Lieb  inequalities \cite{BrascampLieb} (see also \cite{Elliott}) will lead to the result.

In our opinion, the main result of this paper is the description of a simple mechanism for obtaining
exponential relaxation towards equilibrium. One can extend the results to three dimensional momentum preserving collisions, however, so far only for a caricature of Maxwellian molecules. To carry this method
over to the case of hard spheres and for true Maxwellian molecules is an open problem.

The plan of the paper is as follows: In Section~\ref{sec: representation} we derive a representation formula for the Kac evolution $\mathrm{e}^{\mathcal{L}t}$ which is reminiscent of the Ornstein-Uhlenbeck process. This allows us to prove an entropy inequality based upon Nelson's hypercontractive estimate in Section~\ref{sec: hypercontr}. In Section~\ref{sec: correlation} we show how the sharp version of the geometric Brascamp-Lieb inequality leads to a correlation inequality for the entropy involving marginals, which in turn proves our main entropy inequality. The fact that our Brascamp-Lieb datum is geometric relies on a sum rule which will be proved in Section~\ref{sec: proofofsumrule}. A short proof of the geometric form of the Brascamp-Lieb inequalities is deferred to Appendix~\ref{sec: bl}, as well as some technical details to ensure its applicability in Appendix~\ref{sec: approximation}. In Section~\ref{sec: boltzmann} we show how our method can be applied to three-dimensional Maxwellian collisions with a very simple angular dependence.

\medskip
\noindent{\bf Acknowledgements}:  A.G. and T.R. would like to thank Georgia Tech for its hospitality. The work of M.L. and A.G. was supported in part by NSF grant DMS- 1600560 and the Humboldt Foundation. A.G. and T.R. gratefully acknowledge financial support by the Deutsche Forschungsgemeinschaft (DFG) through GRK 1838 (A.G.) and CRC 1173 (T.R.). F.B. gratefully acknowledges financial support from the Simons Foundation award number 359963. T.R. thanks the Karlsruhe House of Young Scientists (KHYS) for a Research Travel Grant supporting the stay at Georgia Tech.


\section{\bf The representation formula} \label{sec: representation}

The aim of this section is to rewrite \eqref{eq:expLt}, that is 
$e^{\mathcal{L}t}F_0$, in a way which is reminiscent of the Ornstein-Uhlenbeck 
process. This representation will naturally lead to the next step in the proof 
of Theorem \ref{thm: main}, namely the entropy inequality that will be presented 
in Theorem \ref{thm: entropypartial}. 

It is convenient to write
\begin{align*}
\mathcal L = \Lambda(Q-I) \ , \text{ where }  \Lambda = \lambda_S \frac{M}{2} + 
\lambda_R \frac{N}{2} + \mu M \ , 
\end{align*}
and the operator $Q$ is a convex combination of $R_{ij}$s, given by 
\begin{align*}
Q =  \frac{\lambda_S}{\Lambda(M-1)} \sum_{1 \leq i<j \leq M} R_{ij}  + 
\frac{\lambda_R}{\Lambda(N-1)} \sum_{M < i<j \leq N + M} R_{ij}  + 
\frac{\mu}{\Lambda N} \sum_{i=1}^{M} \sum_{j=M+1}^{M+N}  R_{ij} \ , 
\end{align*}
i.e., $Q$ is an average over rotation operators.
The right hand side of \eqref{eq:expLt} can be written as
\begin{align} \label{eq:expansion}
(e^{\mathcal L t}F_0)(\mathbf v, \mathbf w) = e^{-\Lambda t} \sum_{k=0}^\infty 
\frac{t^k \Lambda^k}{k!} Q^kF_0(\mathbf v, \mathbf w) \ ,
\end{align}
where
\begin{equation}\label{series}
Q^kF_0(\mathbf v, \mathbf w) =\sum_{\alpha_1, \dots, \alpha_k} 
\lambda_{\alpha_1} \cdots \lambda_{\alpha_k} \
\int_{[-\pi, \pi]^k} \rho(\theta_1) \, \mathrm{d} \theta_1 \cdots 
\rho(\theta_k) 
\mathrm{d} \theta_k \, F_0\left(\left[\prod_{l=1}^k 
r_{\alpha_l}(\theta_l)\right]^{-1} (\mathbf v, \mathbf w)\right) \ .\
\end{equation}

Here, $\alpha$ labels pairs of particles, that is, $\alpha = (i, j)$, $1 \leq i < j \leq M+N$,
$r_{\alpha}(\theta)$ is defined in \eqref{eq:R} and  $\lambda_{\alpha}$ 
is given by the rotation corresponding to the index
$\alpha$, that is,
\begin{align*}
\lambda_{(i,j)}=\frac{\lambda_S}{\Lambda (M-1)} \ &\text{ if }\ 1 \leq i < 
j\leq M\ , \\ 
\lambda_{(i,j)}=\frac{\lambda_R}{\Lambda (N-1)} \ &\text{ if }\ M+1 \leq i < 
j\leq M+N\ , \\ 
\lambda_{(i,j)}=\frac{\mu }{\Lambda N}    \ &\text{ if }\ 1 \leq i  \leq M\ , 
M+1 \leq j\leq M+N\ .
\end{align*}
Note that the sum over {\it all} pairs $\sum_\alpha \lambda_\alpha~=~1$.

For our purpose, it is convenient to write the function $f_0$, introduced in 
\eqref{initialcondition}, as $f_0(\mathbf v) = h_0( \mathbf v)e^{-\pi |\mathbf 
v|^2}$. Since the Gaussian function is invariant under rotations, 
\eqref{eq:expansion} takes the form \begin{equation*} 
(e^{\mathcal L t}F_0)(\mathbf v, \mathbf w) = 
e^{-\pi\left(|\mathbf v|^2+ |\mathbf w|^2\right)}  e^{-\Lambda t} 
\sum_{k=0}^\infty \frac{t^k \Lambda^k}{k!}Q^k \left( h_0\circ P\right)(\mathbf 
v, \mathbf w) \ . 
\end{equation*}
We introduce the projection $P: \R^{N+M} \rightarrow \R^M$ by $P(\mathbf v, 
\mathbf w) = \mathbf v$, as a reminder that the semigroup $e^{\mathcal L t}$ 
acts on functions that depend on $\mathbf v$ as well as $\mathbf w$.  If we 
write
\begin{equation*}
f(\mathbf v,t) = e^{-\pi |\mathbf v|^2}h(\mathbf v,t)\ ,
\end{equation*}
then \eqref{marginal} can be written as
\begin{equation*} 
h(\mathbf v, t) =  e^{-\Lambda t} \sum_{k=0}^\infty \frac{t^k \Lambda^k}{k!} 
h_k(\mathbf v) \ , 
\end{equation*}
where the functions $h_k$ are given by
\begin{equation*}
h_k(\mathbf v) := \int_{\R^N} Q^k \left( h_0\circ P\right)(\mathbf v, \mathbf 
w)e^{-\pi |\mathbf w|^2} \, \mathrm{d} \mathbf w \ . 
\end{equation*}
Likewise, the entropy of $f$ is expressed as
\begin{equation*}
S(f(\cdot, t)) = \int_{\R^M} h(\mathbf v, t) \log h(\mathbf v, t) e^{-\pi 
|\mathbf v|^2} \, \mathrm{d} \mathbf v =: \mathcal S(h(\cdot, t)) \ . 
\end{equation*}
Expanding the function $Q^k ( h_0\circ P)(\mathbf v, \mathbf w)$, we find that 
\begin{multline} \label{h_k}
h_k(\mathbf v)  =\sum_{\alpha_1, \dots, \alpha_k} \lambda_{\alpha_1} \cdots 
\lambda_{\alpha_k} \int_{[-\pi,\pi]^k} \rho(\theta_1) \, \mathrm{d}\theta_1 
\cdots \rho(\theta_k) \mathrm{d} \theta_k \times \\ \times \int_{\R^N} \left(h_0 
\circ P\right)\left(\left[\prod_{l=1}^k r_{\alpha_l}(\theta_l)\right]^{-1} 
(\mathbf v, \mathbf w)\right) e^{-\pi |\mathbf w|^2} \, \mathrm{d} \mathbf w \ ,
\end{multline} 
where, as before, see \eqref{series}, $r_\alpha(\theta)$ rotates the plane given 
by the index pair $\alpha$ by an angle $\theta$ while keeping the other 
directions fixed. Since $P(\mathbf v, \mathbf w)= \mathbf 
v$, it is natural to write
\begin{equation*} 
\left[\prod_{j=1}^k r_{\alpha_j}(\theta_j)\right]^{-1} =  
\begin{pmatrix} 
A_k(\underline \alpha, \underline \theta)  & B_k(\underline \alpha, \underline 
\theta) \\ C_k(\underline \alpha, \underline \theta) &D_k(\underline \alpha, 
\underline \theta) \end{pmatrix}\ , 
\end{equation*}
where $A_k\in\mathbb{R}^{M\times M}$ is an $M\times M$ matrix, 
$B_k\in\mathbb{R}^{M\times N}$, $C_k\in\mathbb{R}^{N \times M}$ and $D_k\in 
\mathbb{R}^{N\times N}$. Further, $\underline \alpha = (\alpha_1, \dots, 
\alpha_k)$ and $\underline \theta = (\theta_1, \dots, \theta_k)$. This notation 
allows us to rewrite \eqref{h_k} as 
\begin{multline*}
h_k(\mathbf v) =\sum_{\alpha_1, \dots, \alpha_k} \lambda_{\alpha_1} \cdots 
\lambda_{\alpha_k} \int_{[-\pi,\pi]^k} \rho(\theta_1) \, \mathrm{d}\theta_1 
\cdots \rho(\theta_k) \, \mathrm{d} \theta_k \times \\ \times \int_{\R^N} 
h_0\left(A_k(\underline \alpha, \underline \theta) \mathbf v + B_k(\underline 
\alpha, \underline \theta) \mathbf w\right) e^{-\pi |\mathbf w|^2} \, 
\mathrm{d}\mathbf w \ .
\end{multline*} 
Note that, by the definition of rotations,
\begin{equation} \label{rottwo}
A_k(\underline \alpha, \underline \theta)A^T_k(\underline \alpha, \underline 
\theta) + B_k(\underline \alpha, \underline \theta)B^T_k(\underline \alpha, 
\underline \theta) = I_M \ . 
\end{equation} 

\begin{lemma} Let $A\in\mathbb{R}^{M\times M}$ and $B\in\mathbb{R}^{M\times N}$ 
be matrices that satisfy $AA^T+BB^T=I_M$. 
Then 
\begin{align*}
\int_{\R^N} h(A \mathbf v + B \mathbf w) e^{-\pi |\mathbf w|^2} \, \mathrm{d} 
\mathbf w \ =\int_{\R^M} h\left(A \mathbf v + (I_M - AA^T)^{1/2} \mathbf 
u\right) e^{-\pi | \mathbf{u}|^2} \, \mathrm{d} \mathbf u 
\end{align*}
for any integrable function $h$.
\end{lemma}

\begin{proof}
Denote the range of $B$ by $H \subset \R^M$ and its kernel by $K \subset \R^N$. 
We may write
\begin{align*}
\int_{\R^N} h(A \mathbf v + B \mathbf w) e^{-\pi |\mathbf w|^2} \, \mathrm{d} 
\mathbf w & = \int_K \int_{K^\perp}  h(A \mathbf v + B \mathbf u) e^{-\pi | 
\mathbf u|^2} e^{-\pi | \mathbf{u}'|^2} \, \mathrm{d} \mathbf u \mathrm{d} 
\mathbf{u}' \\ & = \int_{K^\perp}  h(A \mathbf v + B \mathbf u) e^{-\pi |\mathbf 
u|^2} \, \mathrm{d}\mathbf u  \ .
\end{align*}
The symmetric map $BB^T: \R^M \rightarrow \R^M$ has $H$ as its range and 
$H^\perp$, that is the orthogonal complement of $H$ in $\R^M$, as its kernel. 
Indeed, suppose that there  exists $x \in \R^M$ with $BB^Tx=0$, then $B^Tx = 0$, 
i.e., $x \in {\rm Ker} B^T$ or $x$ is perpendicular to $H$.
Hence, the map $BB^T:H \rightarrow H$ is invertible. Define the linear map $R: 
\R^N \rightarrow H$ by
\begin{align*}
R = \left(BB^T\right)^{-1/2} B
\end{align*}
and note that $RR^T = I_H$ while $R^TR$ projects the space $K^\perp$ 
orthogonally onto $H$. Since $K^\perp$ and $H$ have the same dimension, it 
follows that $R^T$ restricted to $H$ defines an isometry between $H$ and 
$K^\perp$. Hence,
\begin{align*}
\int_{K^\perp}  h\left(A \mathbf v + B \mathbf u\right) e^{-\pi | \mathbf 
u|^2}\, \mathrm{d}\mathbf u & =\int_{K^\perp}  h\left(A \mathbf v + 
\left(BB^T\right)^{1/2}R \mathbf u\right) e^{-\pi |\mathbf u|^2} \, 
\mathrm{d}\mathbf u \\ & = \int_{H}  h\left(A \mathbf v + 
\left(BB^T\right)^{1/2}R R^T \mathbf u\right) e^{-\pi |R^T \mathbf u|^2} \, 
\mathrm{d}\mathbf u \\ & = \int_{H}  h\left(A \mathbf v + 
\left(BB^T\right)^{1/2}\mathbf u\right) e^{-\pi |\mathbf u|^2} \, 
\mathrm{d}\mathbf u \ .
\end{align*}
The assumption $AA^T+BB^T=I_M$, together with the fact that 
\begin{align*}
\int_{H}  h\left(A \mathbf v + \left(BB^T\right)^{1/2}\mathbf u\right) e^{-\pi 
|\mathbf u|^2} \, \mathrm{d}\mathbf u = \int_{H^\perp} \int_{H}  h\left(A 
\mathbf v + (BB^T)^{1/2}\mathbf u\right) e^{-\pi |\mathbf u|^2} \, 
\mathrm{d}\mathbf u \, e^{-\pi |\mathbf u'|^2} \, \mathrm{d} \mathbf u' 
\end{align*}
now implies the lemma. 
\end{proof}

The matrix $A_k(\underline \alpha, \underline \theta) $ has an orthogonal 
singular value decomposition,  
\begin{align} \label{eq:svdA}
A_k(\underline \alpha, \underline \theta) =U_k(\underline \alpha, \underline 
\theta) \Gamma_k (\underline \alpha, \underline \theta) V^T_k(\underline \alpha, 
\underline \theta) \ , 
\end{align}
where $\Gamma_k (\underline \alpha, \underline \theta)={\rm 
diag}[\gamma_{k,1}(\underline \alpha, \underline \theta), \dots, 
\gamma_{k,M}(\underline \alpha, \underline \theta)]$ is the diagonal matrix 
whose entries $\gamma_{k,j}(\underline \alpha, \underline \theta)$, $j=1, \dots, 
M$, are the singular values of $A_k(\underline \alpha, \underline \theta)$, and 
$U_k (\underline \alpha, \underline \theta)$ and $V_k (\underline \alpha, 
\underline \theta)$ are rotations in $\R^M$. Note that \eqref{rottwo} implies 
$\gamma_{k,j}(\underline \alpha, \underline \theta) \in [0,1]$ for $j=1, \dots, 
M$. We shall use the abbreviation
\begin{equation*}
 h_0(U_k(\underline \alpha, \underline 
\theta)\mathbf v)=h_{0, U_k(\underline \alpha, \underline 
\theta)}(\mathbf v) \ .
\end{equation*}
These considerations can be summarized by the 
representation formula presented in the following theorem.

\begin{theorem}[Representation formula] \label{thm: representation}
The function $h_k$ can be written as
\begin{multline} \label{eq: representation}
h_k(\mathbf v)  = \sum_{\alpha_1, \dots, \alpha_k} \lambda_{\alpha_1} \cdots 
\lambda_{\alpha_k} \int_{[-\pi,\pi]^k} \rho(\theta_1)\,  \mathrm{d}\theta_1 
\cdots \rho(\theta_k) \, \mathrm{d} \theta_k \times \\  \times \int_{\R^M}  
h_{0, U_k(\underline \alpha, \underline \theta)} \left(\Gamma_k(\underline 
\alpha, \underline \theta) V_k^T(\underline \alpha, \underline \theta) \mathbf v
+ \left(I_M - \Gamma^2_k(\underline \alpha, \underline \theta) \right)^{1/2} \mathbf w\right) e^{-\pi |\mathbf w|^2}\,  \mathrm{d} \mathbf w \ ,
\end{multline} 
where $\smash{h_{0,U_k(\underline \alpha, \underline \theta)}}$, 
$\Gamma_k(\underline \alpha, \underline \theta)$ and $V_k$ are as defined above. 
\end{theorem}

\section{\bf The hypercontractive estimate } \label{sec: hypercontr}

Starting from \eqref{eq: representation} and using convexity of the entropy 
and Jensen's inequality together with
\begin{equation*}
 \sum_{\alpha_1, \dots, \alpha_k} \lambda_{\alpha_1} \cdots 
\lambda_{\alpha_k} 
\int_{[-\pi,\pi]^k} \rho(\theta_1) \, \mathrm{d}\theta_1 \cdots 
\rho(\theta_k)\,  \mathrm{d} \theta_k =1  \ ,
\end{equation*}
we get
\begin{align*}
\mathcal S(h_k) \le \sum_{\alpha_1, \dots, \alpha_k} \lambda_{\alpha_1} \cdots 
\lambda_{\alpha_k} 
\int_{[-\pi,\pi]^k} \rho(\theta_1) \, \mathrm{d}\theta_1 \cdots 
\rho(\theta_k)\,  \mathrm{d} \theta_k \, \mathcal S(g_k(\cdot,\underline \alpha, 
\underline \theta)),
\end{align*}
where we set
\begin{align}\label{eq:gk}
g_k(\mathbf v,\underline \alpha, 
\underline \theta) = \int_{\R^M}  h_{0, U_k(\underline \alpha, \underline 
\theta)} 
\left(\gamma_k(\underline \alpha, \underline \theta)  \mathbf v
+ \left(I_M - \gamma^2_k(\underline \alpha, \underline \theta) \right)^{1/2} 
\mathbf w\right) e^{-\pi |\mathbf w|^2} \,  \mathrm{d} \mathbf w \ ,
\end{align}
and we removed the rotation $V_k^T(\underline \alpha, \underline \theta)$ by 
a change of variables.

To explain the main observation in this section we look at \eqref{eq:gk} when 
$M=1$. Since $0\leq\gamma_k(\underline \alpha, \underline \theta)\leq 1$, we 
can write $\gamma_k(\underline \alpha, \underline \theta)=e^{-t}$ and we get 
$g_k(v,\underline \alpha, 
\underline \theta)= N_t(h_{0, U_k(\underline \alpha, 
\underline \theta)})$ where $N_t$ is the Ornstein-Uhlenbeck semigroup, that is
\begin{align*}
N_th(x) = \int_\R h\left(e^{-t} x + \sqrt{1-e^{-2t}}y\right) e^{-\pi y^2} \, 
\mathrm{d}y \ .
\end{align*}
Thus Theorem \ref{thm: representation} renders the function $h_k$ as a convex 
combination of terms reminiscent of the Ornstein-Uhlenbeck process, 
albeit in matrix form. We make use of this observation to find a bound for 
$\mathcal S(g_k(\cdot,\underline \alpha, 
\underline \theta))$. This bound together with a suitable correlation 
inequality proved in the next section will lead to a bound for $ \mathcal 
S(h_k)$.

In addition to the notation developed in the 
previous section, we need various marginals of the function 
$\smash{h_{0,U_k(\underline{\alpha}, \underline{\theta})}}$. Quite generally,  
if $h$ is a function of $M$ variables and $\sigma \subset \{1, \dots, M\}$, we 
shall denote by $h^\sigma$ the marginals of $h$ with respect to the variables 
$v_j, j \in \sigma$, for instance,
\begin{align*}
h^{\{1,2\}}(v_3, \dots, v_M) = \int_{\R^2} h(v_1,v_2, v_3, \dots, v_M) 
e^{-\pi\left(v_1^2+v_2^2\right)} \, \mathrm{d}v_1 \mathrm{d}v_2 \ . 
\end{align*}
It will be convenient to use the matrix $P_\sigma: \R^M \rightarrow 
\R^{\vert\sigma\vert}$ that projects $\R^M$ orthogonally onto 
$\R^{\vert\sigma\vert}$ which we will identify with subspace of $\R^M$. To give 
an example, let $\mathbf{v} = (v_1, ..., v_M)$. Then 
$P_{\{1,2\}}\mathbf{v}~=~(v_1, v_2)$.
The following theorem is the main result of this section. 

\begin{theorem}[Partial entropy bound]\label{thm: entropypartial}
Let $h_0 \in L^1(\R^M, e^{-\pi |\mathbf v|^2} d \mathbf v)$ be nonnegative and 
assume that $\mathcal S(h_0) < \infty$. Then 
\begin{multline} \label{est:entropypartial}
\mathcal S(g_k(\cdot,\underline \alpha, 
\underline \theta)) \\ \le 
\sum_{\sigma 
\subset \{1, \dots, M\}} \prod_{i \in \sigma^{\rm c}} \gamma_{k,i}^2 \prod_{j 
\in \sigma} \left(1 - \gamma_{k,j}^2\right) \int_{\R^M} h_0( \mathbf v) \log 
h_{0,U_k(\underline \alpha, \underline \theta)}^\sigma \left(P_{\sigma^{\rm 
c}}U_k(\underline \alpha, \underline \theta)^T\mathbf v\right) e^{-\pi |\mathbf 
v|^2} \, \mathrm{d} \mathbf v \ ,
\end{multline} 
where $\sigma^{\rm c}$ is the complement of the set $\sigma$ in $\{1, ..., M\}$.
\end{theorem}
A key role in the proof of Theorem \ref{thm: entropypartial} is played by 
Nelson's hypercontractive estimate. 
\begin{theorem}[Nelson's hypercontractive estimate] \label{thm:Nelson}
The Ornstein-Uhlenbeck semigroup, 
\begin{align*}
N_th(x) = \int_\R h\left(e^{-t} x + \sqrt{1-e^{-2t}}y\right) e^{-\pi y^2} \, 
\mathrm{d}y \ , 
\end{align*}
for $t\geq0$, is  bounded  from $L^p(\R, e^{-\pi x^2} \, \mathrm{d}x)$ to 
$L^q(\R, e^{-\pi x^2} \, \mathrm{d}x)$ if and only if 
\begin{align*}
(p-1) \ge e^{-2t}(q-1) \ .
\end{align*}
For such values of $p$ and $q$, 
\begin{align*}
\Vert N_th \Vert_q \le \Vert h \Vert_{p}
\end{align*}
with equality if and only if $h$ is constant.
\end{theorem}
\begin{proof} For a proof we refer the reader to \cite{Nelson}. For other proofs 
see \cite{Gross1, Gross2,Federbush,CarlenLoss}. 
\end{proof}

Nelson's hypercontractive estimate, that is Theorem~\ref{thm:Nelson}, implies 
the following Corollary, which will be useful in the proof of Theorem~\ref{thm: 
entropypartial}.

\begin{corollary}[Entropic version of Nelson's hypercontractive estimate] 
\label{entropicnelson} Let $h:\R \rightarrow \R_+$ be a function in $L^1(\R, 
e^{-\pi x^2}\, \mathrm{d}x)$ with finite entropy, i.e.,
\begin{align*}
\mathcal S(h) = \int_\R h(x) \log h(x) \, e^{-\pi x^2} dx < \infty \ .
\end{align*}
Then
\begin{align*}
\mathcal S(N_th) \le e^{-2t} \mathcal S(h) + (1-e^{-2t}) \Vert h \Vert_1 \log 
\Vert h \Vert_1 \  
\end{align*}
for all $t \geq 0$.
\end{corollary}
\begin{proof}
Let $h \in L^p(\R, e^{-\pi x^2} \, \mathrm{d}x)$, for $p \ge 1$ small, be a 
nonnegative function. As $\Vert N_th\Vert_1~=~\Vert h \Vert_1$, we can apply 
Nelson's hypercontractive estimate, which implies that for $p, q$ that satisfy 
$(p-1)~=~e^{-2t}(q-1)$,
\begin{align*}
\frac{\Vert N_th \Vert_q - \Vert N_t h \Vert_1}{q-1} \le \frac{\Vert h \Vert_p 
- \Vert h \Vert_1}{q-1} = e^{-2t} \frac{\Vert h \Vert_p - \Vert h \Vert_1}{p-1} 
\ .
\end{align*}
Sending $p \to 1$ and hence $q\to1$, we get the claimed estimate for such 
functions $h$. If $h$ just has finite entropy one cuts off $h$ at large values, 
uses the above estimate and removes the cutoff using the monotone convergence 
theorem.
\end{proof}
We are now ready to prove Theorem~\ref{thm: entropypartial}.
\begin{proof}[Proof of Theorem \ref{thm: entropypartial}]
Remember that $0 \le \gamma_{k,j}(\underline{\alpha},\underline{\theta}) \le 1$ 
for $j = 1,...,M$. Thus, by inductively applying Corollary \ref{entropicnelson} 
to
\begin{align*}
 \int_{\R^M}h_{0, U_k(\underline \alpha, \underline \theta)} \left( 
\gamma_{k,1} v_1 + \sqrt{1 - \gamma_{k,1}^2} \, u_1, \dots, \gamma_{k,M} v_M + 
\sqrt{1 -\gamma_{k,M}^2} \, u_M\right) e^{-\pi \sum_{j=1}^M u_j^2} \, 
\mathrm{d}u_1 \cdots \mathrm{d} u_M \ ,
\end{align*}
we obtain
\begin{multline*}
\mathcal S(g_k(\cdot,\underline \alpha, 
\underline \theta)) \le 
\sum_{\sigma \subset \{1, \dots, M\}} \prod_{i \in \sigma^{\rm c}} 
\gamma_{k,i}^2 \prod_{j \in \sigma} (1 - \gamma_{k,j}^2) \int_{\R^{|\sigma^{\rm 
c}|}} h_{0,U_k(\underline \alpha, \underline \theta)}^\sigma(\mathbf u) \log 
h_{0,U_k(\underline \alpha, \underline \theta)}^\sigma(\mathbf u) \, e^{-\pi 
|\mathbf u|^2} \, \mathrm{d} \mathbf u  \ .
\end{multline*}
Inserting the definition of the marginal $\smash{h_{0,U_k(\underline \alpha, 
\underline \theta)}^\sigma}$, we see that  
\begin{align*}
\int_{\R^{|\sigma^{\rm c}|}} h_{0,U_k(\underline \alpha, \underline 
\theta)}^\sigma(\mathbf u) \log h_{0,U_k(\underline \alpha, \underline 
\theta)}^\sigma(\mathbf u) \, e^{-\pi |\mathbf u|^2} \, \mathrm{d} \mathbf u
& =\int_{\R^M} h_{0,U_k(\underline \alpha, \underline 
\theta)}^\sigma(P_{\sigma^{\rm c}} \mathbf v) \log h_{0,U_k(\underline \alpha, 
\underline \theta)}^\sigma(P_{\sigma^{\rm c}}\mathbf v) \, e^{-\pi |\mathbf 
v|^2} \, \mathrm{d} \mathbf v \\
& = \int_{\R^M} h_{0,U_k(\underline \alpha, \underline \theta)}( \mathbf v) 
\log h_{0,U_k(\underline \alpha, \underline \theta)}^\sigma(P_{\sigma^{\rm 
c}}\mathbf v) \, e^{-\pi |\mathbf v|^2} \, \mathrm{d} \mathbf v \\
& = \int_{\R^M} h_0( \mathbf v) \log h_{0,U_k(\underline \alpha, \underline  
\theta)}^\sigma(P_{\sigma^{\rm c}}U_k(\underline \alpha, \underline 
\theta)^T\mathbf v) \, e^{-\pi |\mathbf v|^2} \, \mathrm{d} \mathbf v  ,
\end{align*}
which finishes the proof of Theorem \ref{thm: entropypartial}.
\end{proof}

\section{\bf The key entropy bound}\label{sec: correlation}

Collecting the results of the previous sections we get the following bound
\begin{align}\label{est:partial}
\mathcal S(h_k) \leq & \sum_{\alpha_1, \dots, \alpha_k} \lambda_{\alpha_1} 
\cdots \lambda_{\alpha_k} 
\int_{[-\pi,\pi]^k} \rho(\theta_1) \, \mathrm{d}\theta_1 \cdots \rho(\theta_k) 
\, \mathrm{d} \theta_k \, \times \notag \\
&\times \sum_{\sigma \subset \{1, \dots, M\}} \prod_{i \in 
\sigma^{\rm c}} \gamma_{k,i}^2 \prod_{j \in \sigma} \left(1 - 
\gamma_{k,j}^2\right) \int_{\R^M} h_0( \mathbf v) \log h_{0,U_k(\underline 
\alpha, \underline \theta)}^\sigma \left(P_{\sigma^{\rm c}}U_k(\underline 
\alpha, \underline \theta)^T\mathbf v\right) \, e^{-\pi |\mathbf v|^2} \, 
\mathrm{d} \mathbf v .
\end{align}
The right-hand side of  \eqref{est:partial} contains a large sum over the entropy of 
marginals of $h_0$. In order to bound such a sum in terms of the entropy of 
$h_0$ one may try to apply some version of the Loomis-Whitney inequality 
\cite{LW} or, more precisely, of an inequality by Han \cite{Han}. This is 
essentially correct, but will require a substantial generalization of this inequality.
 Let us first formulate the main theorem of this section.

\begin{theorem}[Entropy bound] \label{thm: entropyestimate}
The estimate
\begin{align} \label{eq:estimatemarginals}
{\mathcal S}(h_k)\le \left[\frac{M}{N+M} 
+\frac{N}{N+M}\left(1-\mu_\rho\frac{N+M}{N\Lambda}\right)^k\right] 
\mathcal S(h_0) 
\end{align}
holds. 
\end{theorem}
As mentioned before, to prove 
Theorem~\ref{thm: entropyestimate}, we need a generalized version of an 
inequality by Han. This generalization was proven by Carlen-Cordero-Erausquin 
in \cite{CarlenCordero}. It is based on the geometric Brascamp-Lieb inequality 
due 
to Ball \cite{ball1}, see also \cite{ball2}, in the rank one case, and due to 
Barthe \cite{Barthe1} in the general case.
\begin{theorem}[Correlation inequality] \label{thm: correlation}
For $i=1, \dots K$,  let $H_i \subset \R^M$ be subspaces of dimension $d_i$ and 
$B_i: \R^M \rightarrow H_i$ be linear maps with the property that $B_i B_i^T= 
I_{H_i}$, the identity map on $H_i$.  Assume further that there are non-negative 
constants $c_i, i=1, \dots, K$ such that
\begin{equation} \label{eq: sumrule}
\sum_{i=1}^K c_i B_i^T B_i = I_M \ .
\end{equation}
Then, for nonnegative functions $f_i: H_i \rightarrow \R$,
\begin{equation} \label{eq: brascamplieb}
\int_{\R^M} \prod_{i=1}^K f_i^{c_i} (B_i \mathbf v) \, e^{-\pi |\mathbf v|^2}  
\, \mathrm{d} \mathbf v \le \prod_{i=1}^K \left( \int_{H_i} f_i( u) \, 
e^{-\pi |u|^2} \, \mathrm{d} u \right)^{c_i}  \ .
\end{equation}
Moreover, 
\begin{equation} \label{eq: entropymarginalstwo}
\int_{\R^M} h(\mathbf v) \log h(\mathbf v) \, e^{-\pi |\mathbf v|^2} \, \mathrm{d} \mathbf v
\ge \sum_{i=1}^K c_i \left[ \int_{\R^M} h(\mathbf v) \log f_i(B_i 
\mathbf v) \, e^{-\pi |\mathbf v|^2} \, \mathrm{d} \mathbf v- \log \int_{H_i} 
f_i( u) \, e^{-\pi {| u|}^2} \, \mathrm{d} u  
\right] , 
\end{equation}
for any nonnegative function $h \in  L^1(\R^M, e^{-\pi |\mathbf v|^2} d \mathbf 
v)$. 
\end{theorem}
Since Theorem \ref{thm: correlation} is very useful  in a number of 
applications, and for the readers convenience, we will give an elementary proof 
in Appendix \ref{sec: bl}.  
\begin{remark} By taking the trace in \eqref{eq: sumrule} one sees that
\begin{align*}
\sum_{i=1}^K c_i d_i = M \ .
\end{align*}
\end{remark}

We would like to apply \eqref{eq: entropymarginalstwo} to the right hand side 
of \eqref{est:partial}. An immediate problem is that \eqref{est:partial} is in terms of integrals and not sums. While 
there are some results available for continuous indices (see, e.g., 
\cite{Barthe2}), they do not apply to our situation and hence we will take a 
more direct approach and approximate
the measure $\rho(\theta) \mathrm{d} \theta$ by a discrete measure. It is 
important that the approximation also satisfies the constraint 
\eqref{assumprho}. The following lemma establishes such an approximation. Its 
proof is given in Appendix \ref{sec: approximation}.

\begin{lemma} \label{lem: approximation} 
Let $\rho$ be a probability density on $[-\pi, \pi]$ whose Fourier series 
converges absolutely and assume that \eqref{assumprho} is satisfied. There 
exists a sequence of discrete probability measures $\nu_K$, $K =1,2, \dots$, 
such that for every continuous function $f$ on $[-\pi,\pi]$
\begin{align*}
\lim_{K \to \infty} \int_{-\pi}^\pi f(\theta) \, \nu_K(\mathrm{d} \theta) = 
\int_{-\pi}^\pi f(\theta) \rho(\theta) \, \mathrm{d} \theta \ . 
\end{align*}
Moreover, 
\begin{align*}
\int_{-\pi}^\pi \cos \theta \sin \theta \, \nu_K(\mathrm{d} \theta)  = 0 \ , 
\end{align*}
for all $K \in \N$. More precisely, 
\begin{align*}
\nu_K( \mathrm{d} \theta) = \frac{2\pi}{4K+1} \sum_{\ell = -2K}^{2K} 
\rho_K\left(\frac{2 \pi \ell}{4K+1}\right) \delta\left(\theta - \frac{2\pi 
\ell}{4K+1}\right) \, \mathrm{d} \theta \ , 
\end{align*}
where
\begin{align*}
\rho_K(\theta) = \int_{-\pi}^\pi \rho(\theta -\phi)\ p_K(\theta)
\, \mathrm{d} \phi
\ \text{ and } \ p_K(\theta) := \frac{1}{2K+1} \left( \sum_{k= - K}^K e^{i 
k\theta}\right)^2 \ .
\end{align*}
\end{lemma}

At this point we can prepare the ground for the application of Theorem 
\ref{thm: correlation} to inequality \eqref{est:partial}. We first replace 
$\rho(\theta)\mathrm{d} 
\theta$ in \eqref{est:partial} with $\nu_K(\mathrm{d} \theta)$. Setting
\begin{equation*} \label{convention}
\omega_{\ell_j}=\rho_K( \theta_j)  \ , \ \theta_{\ell_j} =\frac{2\pi \ell_j 
}{4K+1} \ , \ \text{and } \underline \theta = (\theta_{\ell_1}, \dots, 
\theta_{\ell_k}) \ ,
\end{equation*}
we obtain
\begin{align}
& \sum_{\alpha_1, \dots, \alpha_k} \lambda_{\alpha_1} \cdots \lambda_{\alpha_k} 
\int_{[-\pi,\pi]^k} \nu_K(\mathrm{d} \theta_1) \cdots \nu_K(\mathrm{d} \theta_k) 
\sum_{\sigma \subset \{1, \dots, M\}} \prod_{i \in \sigma^{\rm c}} 
\gamma_{k,i}(\underline \alpha, \underline \theta)^2 \prod_{j \in \sigma} 
\left(1 - \gamma_{k,j}(\underline \alpha, \underline \theta)^2\right) \times 
\nonumber\\ 
& \qquad\qquad\qquad \times \int_{\R^M} h_0( \mathbf v) \log h_{0,U_k(\underline 
\alpha, \underline \theta)}^\sigma(P_{\sigma^{\rm c}}U_k(\underline \alpha, 
\underline \theta)^T\mathbf v) \, e^{-\pi |\mathbf v|^2} \, \mathrm{d} \mathbf v 
\nonumber\\  
& = \sum_{\alpha_1, \dots, \alpha_k} \lambda_{\alpha_1} \cdots 
\lambda_{\alpha_k} \sum_{-K \le \ell_1, \dots, \ell_k \le K } \prod_{j=1}^k 
\omega_{\ell_j}  \sum_{\sigma \subset \{1, \dots, M\}}   \prod_{i \in 
\sigma^{\rm c}} \gamma_{k,i}(\underline \alpha, \underline \theta)^2 \prod_{j 
\in \sigma} \left(1 - \gamma_{k,j}(\underline \alpha, \underline 
\theta)^2\right) \times \nonumber\\ 
& \qquad\qquad\qquad \times \int_{\R^M} h_0( \mathbf v) \log h_{0,U_k(\underline 
\alpha, \underline \theta)}^\sigma(P_{\sigma^{\rm c}}U_k(\underline \alpha, 
\underline \theta)^T\mathbf v) \, e^{-\pi |\mathbf v|^2} \, \mathrm{d} \mathbf v 
 \ . \label{eq:discretesum}
\end{align}

In order to apply Theorem \ref{thm: correlation} to \eqref{eq:discretesum} we 
have to replace the sum over the index~$i$ with a sum over the indices 
$\alpha_1, \dots, \alpha_k, \ell_1, \dots \ell_k$ and all subsets $\sigma 
\subset \{1, \dots, M\}$. Moreover, we substitute
\begin{align*}
& \text{the constants } c_i & & \text{ by }  \frac{1}{C_{k,M}}\lambda_{\alpha_1} 
\cdots \lambda_{\alpha_k}  \prod_{j=1}^k \omega_{\ell_j} \prod_{i \in 
\sigma^{\rm c}} \gamma_{k,i}(\underline \alpha, \underline \theta)^2 \prod_{j 
\in \sigma} (1 - \gamma_{k,j}(\underline \alpha, \underline \theta)^2) \ , \\ 
& \text{the functions } \vphantom{\prod_{j \in \sigma}} f_i (\mathbf w) & & 
\text{ by }  h_{0,U_k(\underline \alpha, \underline \theta)}^\sigma( \mathbf w)  
\ , \\ 
& \text{the linear maps } \vphantom{\prod_{j \in \sigma}} B_i & & \text{ by }  
P_{\sigma^{\rm c}}U_k(\underline \alpha, \underline \theta)^T \ , \\ 
& \text{the functions } \vphantom{\prod_{j \in \sigma}} f_i(B_i\mathbf v) & & 
\text{ by } h_{0,U_k(\underline \alpha, \underline 
\theta)}^\sigma(P_{\sigma^{\rm c}}U_k(\underline \alpha, \underline 
\theta)^T\mathbf v) \ , \\ 
& \text{and the subspaces } \vphantom{\prod_{j \in \sigma}}  H_i & & \text{ by } 
\R^{\vert\sigma^c\vert}  \ . 
\end{align*}
For any given index $i$ the condition $B_i B_i^T =I_{H_i}$ corresponds to 
$P_{\sigma^{\rm c}}U_k(\underline \alpha, \underline \theta)^T U_k(\underline 
\alpha, \underline \theta) P_{\sigma^{\rm c}}~=~P_{\sigma^{\rm c}}$ which is the 
identity on $\R^{\vert\sigma^{\rm c}\vert}$. 

The next theorem establishes the sum rule \eqref{eq: sumrule} in our setting  
and hence ensures the applicability of Theorem~\ref{thm: correlation} to 
\eqref{eq:discretesum}.
\begin{theorem}[The sum rule] \label{thm: Sumrule} 
If $\nu(\mathrm d \theta)$ is a probability measure satisfying 
\eqref{assumprho}, then
\begin{equation} \label{eq: sumruletwo}
\begin{multlined} 
\sum_{\alpha_1, \dots, \alpha_k} \lambda_{\alpha_1}\cdots \lambda_{\alpha_k} 
\int_{[-\pi,\pi]^k}  \nu(\mathrm{d} \theta_1)\, \cdots \nu( \mathrm{d} 
\theta_k)  \times \\
\times \sum_{\sigma \subset \{1, \dots, M\}} \prod_{i 
\in \sigma^{\rm c}} \gamma_{k,i}(\underline \alpha, \underline \theta)^2 
\prod_{j \in \sigma} \left(1 - \gamma_{k,j}(\underline \alpha, \underline 
\theta)^2\right)   U_k(\underline \alpha, \underline \theta)P_{\sigma^{\rm 
c}}^T 
P_{\sigma^{\rm c}}U_k(\underline \alpha, \underline \theta)^T = C_{k,M} I_M   
\ , 
\end{multlined}
\end{equation}
where 
\begin{align*}
C_{k,M} = \left[ \frac{M}{N+M} 
+\frac{N}{N+M}\left(1-\mu_\nu\frac{N+M}{N\Lambda}\right)^k\right] 
\end{align*} 
with
\[
 \mu_\nu=\mu \int \nu(\mathrm d \theta)\sin^2\theta \ .
\]

\end{theorem}

The proof will be given in Section \ref{sec: proofofsumrule}. We observe here 
that it follows from Theorem \ref{lem: approximation}  that 
$\mu_\rho=\lim_{K\to\infty}\mu_{\nu_K}$.

\begin{proof}[Proof of Theorem \ref{thm: entropyestimate} ]
First we consider the case where $\rho$ is repaced by $\nu_K$ and use Theorem 
\ref{thm: correlation} together with Theorem \ref{thm: Sumrule} and the 
identification rules described above. 
The entropy inequality \eqref{eq: entropymarginalstwo} now says that 
\begin{align*}
& \int_{\R^M} h_0(\mathbf v) \log h_0(\mathbf v) e^{-\pi |\mathbf v|^2} \, 
\mathrm{d} \mathbf v  \\ 
& \ge  \frac{1}{C_{k,M}}  \sum_{\alpha_1, \dots, \alpha_k} \lambda_{\alpha_1} 
\cdots \lambda_{\alpha_k} \sum_{-K \le \ell_1, \dots, \ell_k \le K } 
\prod_{j=1}^k \omega_{\ell_j} \sum_{\sigma \subset \{1, \dots, M\}} \prod_{i \in 
\sigma^{\rm c}} \gamma_{k,i}(\underline \alpha, \underline \theta)^2 \prod_{j 
\in \sigma}  \left(1 - \gamma_{k,j}(\underline \alpha, \underline 
\theta)^2\right) \times  \\ 
& \qquad \times \left[ \int_{\R^M} h_0( \mathbf v) \log h_{0,U_k(\underline 
\alpha, \underline \theta)}^\sigma(P_{\sigma^{\rm c}}U_k(\underline \alpha, 
\underline \theta)^T\mathbf v) \, e^{-\pi |\mathbf v|^2} \, \mathrm{d} \mathbf v 
- \log \int_{\R^{|\sigma^{\rm c}|}}  
h_{0,U_k(\underline \alpha, \underline \theta)}^\sigma( u) \, e^{-\pi 
| u|^2} \, \mathrm{d} u \right] \ . 
\end{align*}
However, since $h_0$ is normalized and $U_k(\underline \alpha, \underline 
\theta)$ is orthogonal, we find that 
\begin{align*}
\int_{\R^{\vert\sigma^{\rm c}\vert}} h_{0,U_k(\underline \alpha, \underline 
\theta)}^\sigma( u) \, e^{-\pi |u|^2} \, \mathrm{d} \ u & = 
\int_{\R^{\vert\sigma^{\rm c}\vert}} \int_{\R^{\vert\sigma\vert}} 
h_{0,U_k(\underline \alpha, \underline \theta)}( v,  u) \, 
 e^{-\pi | v|^2} \, \mathrm{d}  v  \, e^{-\pi |u|^2} \, 
\mathrm{d} u \\  
& = \int_{\R^M} h_0(U_k(\underline \alpha, \underline \theta)  \mathbf v) \, 
e^{-\pi| \mathbf v|^2} \, \mathrm{d}  \mathbf v \\ 
& = 1 \ .
\end{align*}
Thus we find
\begin{align}\label{est:descrete}
 \sum_{\alpha_1, \dots, \alpha_k} \lambda_{\alpha_1} \cdots 
\lambda_{\alpha_k} \sum_{-K \le \ell_1, \dots, \ell_k \le K } \prod_{j=1}^k 
\omega_{\ell_j}  \sum_{\sigma \subset \{1, \dots, M\}}   \prod_{i \in 
\sigma^{\rm c}} \gamma_{k,i}(\underline \alpha, \underline \theta)^2 \prod_{j 
\in \sigma} \left(1 - \gamma_{k,j}(\underline \alpha, \underline 
\theta)^2\right) \times & \\ 
 \times \int_{\R^M} h_0( \mathbf v) \log 
h_{0,U_k(\underline 
\alpha, \underline \theta)}^\sigma(P_{\sigma^{\rm c}}U_k(\underline \alpha, 
\underline \theta)^T\mathbf v) \, e^{-\pi |\mathbf v|^2} \, \mathrm{d} \mathbf 
v \,  \leq & \, C_{k,M}\mathcal S(h_0)\ .\nonumber
\end{align}
As $K \to \infty$, the left-hand side of \eqref{est:descrete} converges to the 
right-hand side of \eqref{est:partial}. 
\end{proof}

We now have all ingredients to give the proof of Theorem~\ref{thm: main}.

\begin{proof}[Proof of Theorem   \ref{thm: main}]
Recall from  Section~\ref{sec: representation}, that 
\begin{align*}
f(\mathbf v, t) = e^{-\pi \vert \mathbf v \vert^2}e^{-\Lambda t} \sum_{k = 0}^{\infty} \frac{t^k \Lambda^k}{k!} h_k(\mathbf v)\ ,
\end{align*}
and that $S(f(\cdot, t)) = \mathcal{S}(h(\cdot, t))$. Combining Theorem~\ref{thm: entropypartial} and Theorem~\ref{thm: entropyestimate}, we obtain
\begin{align*}
\mathcal{S}(h_k) \le C_{k,M}
\mathcal S(h_0) \ ,
\end{align*}
and computing 
\begin{align*}
e^{-\Lambda t} \sum_{k=0}^\infty \frac{\Lambda^k t^k}{k!} C_{k,M}
\end{align*}
yields Theorem~\ref{thm: main}.
\end{proof}

\section{\bf The sum rule. Proof of Theorem \ref{thm: Sumrule}}  \label{sec: proofofsumrule}
We have to compute the matrix
\begin{multline*}
Z := \sum_{\alpha_1, \dots, \alpha_k} \lambda_{\alpha_1} \cdots 
\lambda_{\alpha_k} \int_{[-\pi,\pi]^k} \nu(\mathrm{d} \theta_1) 
\, \cdots \nu(\mathrm{d} \theta_k) \times \\ 
\times \sum_{\sigma \subset \{1, \dots, M\}} \prod_{i \in \sigma^{\rm c}} 
\gamma_{k,i}(\underline \alpha, \underline \theta)^2 \prod_{j \in \sigma} 
\left(1 - \gamma_{k,j}(\underline \alpha, \underline \theta)^2\right)  
U_k(\underline \alpha, \underline \theta)P_{\sigma^{\rm c}}^T P_{\sigma^{\rm 
c}}U_k(\underline \alpha, \underline \theta)^T \ . 
\end{multline*}
Obviously $P_{\sigma^{\rm c}}^T P_{\sigma^{\rm c}} = P_{\sigma^{\rm c}}$ and hence
\begin{multline*}
Z = \sum_{\alpha_1, \dots, \alpha_k} \lambda_{\alpha_1} \cdots 
\lambda_{\alpha_k}  \int_{[-\pi,\pi]^k}  \nu(\mathrm{d} \theta_1)  \cdots 
\nu(\mathrm{d} \theta_k)  \, \times \\ 
\times U_k(\underline \alpha, \underline \theta) \left[ \sum_{\sigma \subset 
\{1, \dots, M\}}  \prod_{i \in \sigma^{\rm c}}\gamma_{k,i}(\underline \alpha, 
\underline \theta)^2 \prod_{j \in \sigma} \left(1 - \gamma_{k,j}(\underline 
\alpha, \underline \theta)^2\right)P_{\sigma^{\rm c}}\right] U_k(\underline 
\alpha, \underline \theta)^T 
\ .
\end{multline*} 
The sum on $\sigma$ is easily evaluated and yields the matrix 
$\Gamma_k^2(\underline \alpha, \underline \theta)$. Hence, recalling the 
orthogonal singular value decomposition \eqref{eq:svdA} of $A_k(\underline 
\alpha, \underline \theta)$, that is, $A_k(\underline \alpha, \underline \theta) 
=U_k(\underline \alpha, \underline \theta) \Gamma_k (\underline \alpha, 
\underline \theta) V^T_k(\underline \alpha, \underline \theta) $, we find that
\begin{equation} \label{average}
Z =  \sum_{\alpha_1, \dots, \alpha_k} \lambda_{\alpha_1} \cdots 
\lambda_{\alpha_k} \int_{[-\pi,\pi]^k} \nu(\mathrm{d} \theta_1)  \cdots 
\nu(\mathrm{d} \theta_k) \, A_k(\underline \alpha, \underline 
\theta)A^T_k(\underline \alpha, \underline \theta) \ .
\end{equation}
One can think about this expression in the following fashion. Recall that
\begin{align*}
\left[\prod_{l=1}^k r_{\alpha_l}(\theta_l)\right]^{-1} = 
\begin{pmatrix}  
A_k(\underline \alpha, \underline \theta) &  B_k(\underline \alpha, \underline 
\theta) \\  C_k(\underline \alpha, \underline \theta) &  D_k(\underline \alpha, 
\underline \theta)\end{pmatrix} \ . 
\end{align*}
With this notation, the matrix $Z$ equals the top left entry of the matrix
\begin{align*}
\sum_{\alpha_1, \dots, \alpha_k} \lambda_{\alpha_1} \cdots 
\lambda_{\alpha_k} \int_{[-\pi,\pi]^k} \nu(\mathrm{d} \theta_1)  \cdots 
\nu(\mathrm{d} \theta_k) \, \left[\prod_{l=1}^k 
r_{\alpha_l}(\theta_l)\right]^{-1}  \begin{pmatrix} I_M & 0 \\ 0 & 
0\end{pmatrix} \left[\prod_{l=1}^k r_{\alpha_l}(\theta_l)\right] \ . 
\end{align*}
The computation hinges on a repeated application of the elementary identity 
\begin{align*}
& \int_{-\pi}^{\pi} \nu(\mathrm{d}\theta) \,  
\begin{pmatrix}\cos(\theta) & -\sin(\theta)\\ \sin(\theta)  & 
\cos(\theta)\end{pmatrix}
\begin{pmatrix} m_1 & 0\\ 0 & m_2\end{pmatrix}
\begin{pmatrix} \cos(\theta) & \sin(\theta)\\ -\sin(\theta) & 
\cos(\theta)\end{pmatrix} \\
& = \begin{pmatrix} (1-\tilde\nu) m_1+\tilde \nu m_2 & 0\\
0 & (1-\tilde\nu) m_2+\tilde \nu m_1 \end{pmatrix} \ ,
\end{align*}
where $\tilde\nu=\int \nu(\mathrm{d} \theta) \, \sin^2(\theta).$
For this to be true we just need \eqref{assumprho}.
We easily check that for the 
rotations $r_\alpha(\theta)$ 
\begin{align} \label{eq:proofsumrule}
& \sum_{\alpha} \lambda_{\alpha} \int_{-\pi}^{\pi} \nu(\mathrm{d} \theta) \, 
r_\alpha(\theta)^{-1}  \begin{pmatrix}  m_1 I_M & 0 \\ 0 &  
m_2 I_N\end{pmatrix}   r_\alpha(\theta)\notag \\ 
& = \frac{1}{\Lambda}\left(\frac{M\lambda_S}{2}+\frac{N\lambda_R}{2}\right)
\begin{pmatrix}  m_1 I_M & 0 \\ 0 &  
m_2 I_N\end{pmatrix} \notag \\
& \qquad + \frac{\mu}{\Lambda N}\begin{pmatrix}   
N(M-1)+N((1-\tilde\nu) m_1 + \tilde\nu m_2)I_M & 0 \\0 & (N-1)M+ 
 M(\tilde\nu m_1 +(1-\tilde\nu)m_2) I_N \end{pmatrix} \notag 
\\
& = \begin{pmatrix}  m_1 I_M & 0 \\ 0 &  
m_2 I_N\end{pmatrix}+ \frac{\mu_\nu}{\Lambda N}
\begin{pmatrix}  
N( m_2-m_1)I_M & 0 \\
0 &  M(m_1 -m_2)I_N
    \end{pmatrix} \ .
\end{align}
where $\mu_\nu=\tilde \nu \mu$. Denote by $L(\nu_1,\nu_2)$  the $(N+M)\times(N+M)$ matrix  
\begin{align*}
L(m_1,m_2)=\begin{pmatrix}  m_1 I_{M} & 0\\ 0 & 
m_2 I_N\end{pmatrix} \ ,
\end{align*}
and set 
\begin{align*}
\mathcal P=I_2 - \frac{\mu_\nu}{\Lambda N}\begin{pmatrix} 
                                            N & -N\\
                                            -M & M
                                           \end{pmatrix} \ .
\end{align*}
Then  \eqref{eq:proofsumrule} is recast as
\begin{align} \label{eq: proofsumrulereduced}
 \sum_{\alpha} \lambda_{\alpha} \int_{-\pi}^{\pi} \nu(\mathrm{d} \theta) \, 
r_\alpha(\theta)^{-1}L(m_1,m_2)r_\alpha(\theta)= L(m_1',m_2') \ ,
\end{align}
where 
\begin{align*}
\begin{pmatrix}  m_1' \\ m_2'\end{pmatrix}=
\mathcal P\begin{pmatrix}  m_1 \\ m_2\end{pmatrix} \ .
\end{align*}
By a repeated application of \eqref{eq: proofsumrulereduced} we obtain
\begin{align*}
& \sum_{\alpha_1, \dots, \alpha_k} \lambda_{\alpha_1} 
\cdots \lambda_{\alpha_k} \int_{[-\pi,\pi]^k} \nu(\mathrm{d}\theta_1) \, 
\cdots \nu(\mathrm{d}\theta_k) \, \left[\prod_{j=1}^k r_{\alpha_j}(\theta_j)\right]^T L(\underline m)\left[\prod_{j=1}^k r_{\alpha_j}(\theta_j)\right] = \vphantom{\sum_{\alpha_j}} L(\mathcal P^k\underline m) \ .
\end{align*}
Thus,
\[
Z =\left(\mathcal P^k\, \begin{pmatrix}  1 \\ 
0\end{pmatrix}\right)_1 I_M \ .
\]
It is easy to see that $\mathcal P$ has eigenvalues $\ell_1=1$ and 
$\ell_2=1-\mu_\nu (M+N)/(\Lambda N)$ with eigenvectors $\underline 
m_1=(1, 1)$ and $\underline m_2=(N,-M)^T/(M+N)$. Consequently, 
\[
 \begin{pmatrix}  1 \\ 0\end{pmatrix}=\frac{M}{N+M}\underline 
m_1+\underline m_2 \ ,
\]
which yields
\[
 \left(\mathcal P^k\, \begin{pmatrix}  1 \\ 
0\end{pmatrix}\right)_1=\frac{M}{N+M}+ \frac{N}{M+N}\left(1-\mu_\nu
\frac{M+N}{\Lambda N}\right)^k \ .
\]
This proves Theorem \ref{thm: Sumrule}. \hfill $\qed$

\section{\bf Boltzmann-Kac collisions} \label{sec: boltzmann}
In this section we show that the above results can also be extended, at least in 
a particular case, to three-dimensional Boltzmann-Kac collisions. 

Again we consider a system of $M$ particles coupled to a reservoir 
consisting of $N$ particles, but now with velocities 
$v_1,\dots,v_M$, $w_1,\dots,w_N\in\R^3$. The collisions between a pair of 
particles have to conserve energy and momentum, 
\begin{align*}
	z_i^2 + z_j^2 &= (z_i^*)^2 + (z_j^*)^2 \\
	z_i + z_j &= z_i^* + z_j^* \ ,
\end{align*}
where $z$ can be either the velocity of a system particle $v$ or of a reservoir 
particle $w$. A convenient parametrization of the post-collisional velocities in 
terms of the velocities before the collision is given by 
\begin{align*}
	z_i^*(\omega) &= z_i -  \omega \cdot(z_i - z_j)  \, \omega\\
	z_j^*(\omega) &= z_j +   \omega  \cdot(z_i - z_j) \, \omega , \quad 
\text{where } \omega \in \mathbb{S}^2 \ .
\end{align*}
This is the so-called $\omega$-representation.  This representation is 
particularly useful, because the velocities are related to each other by a 
\emph{linear} transformation, and the strategy used to proof the results for  
the one-dimensional Kac system carries over rather directly. The direction 
$\omega$ will be chosen according to the uniform probability distribution on 
the unit sphere $\mathbb{S}^2$. 

Introduce the operators 
\begin{align*}
	(R_{ij} f)(\bm{z}) = \int_{\mathbb{S}^2} f(r_{ij}(\omega)^{-1}\bm{z}) \,\mathrm{d}\omega \ ,
\end{align*}
where $\mathrm{d}\omega$ denotes the uniform probability measure on the sphere 
and the matrices $r_{ij}(\omega)$ are symmetric involutions acting as
\begin{align*}
	\begin{pmatrix}
		z_i^*\\ z_j^*
	\end{pmatrix} 
	= \begin{pmatrix}
		I-\omega \omega^T & \omega \omega^T \\ \omega \omega^T & I-\omega \omega^T
	\end{pmatrix}
	\begin{pmatrix}
		z_i \\ z_j
	\end{pmatrix}
\end{align*} 
on the velocities of the particles $i$ and $j$, and as identities otherwise.  
They will replace the one-dimensional Kac collision operators in \eqref{eq:model} 
in the otherwise unchanged generator of the time evolution. Notice that the 
matrices $r_{ij}(\omega)$ are orthogonal, so that the expansion formula 
\eqref{series} still holds with the obvious changes in the dimension of the 
single-particle spaces. 

We prove an analog of Theorem \ref{thm: main} for the case of three-dimensional 
Boltzmann-Kac collisions and pseudo-Maxwellian molecules. 
\begin{theorem}\label{thm: boltzmann-main}
	Let $N\geq M$ and $F_0(\bm{v}, \bm{w}) = 
	f_0(\bm{v})\,\mathrm{e}^{-\pi|\bm{w}|^2}$ for some probability 
	distribution $f_0$ on $\R^{3M}$.  Then the entropy of the marginal 
	\begin{align*}
		f(\bm{v}, t):= \int_{\R^{3N}} \left(\mathrm{e}^{\mathcal{L}t}F_0\right)(\bm{v}, \bm{w})\,\mathrm{d}\bm{w}
	\end{align*}
	with respect to the thermal state $\mathrm{e}^{-\pi|\bm{v}|^2}$ is bounded by
	\begin{align*}
		S(f(\cdot,t)) \leq \left[\frac{N}{N+M} + \frac{N}{N+M} \mathrm{e}^{-\frac{\mu}{3} \frac{N+M}{N} t} \right]\, S(f_0)\ .
	\end{align*}
\end{theorem}

\begin{remark}
	The result in three dimensions is very similar to the case of 
	one-dimensional Kac collisions, with the difference that the rate of  
	exponential decay is $\mu/3$ instead of $\mu_{\rho}$. The factor $1/3$ 
	comes from the fact that $\int_{\mathbb{S}^2} \mathrm{d}\omega \, 
	\omega\omega^T = I_{3}/3$. 
	It would be interesting to cover the true Maxwellian molecules interaction 
	\begin{align*}
		(R_{ij} f)(z) =  \int_{\Sp^2} b\left( \frac{v_i-v_j}{|v_i-v_j|} \cdot \omega \right) f(r_{ij}(\omega)^{-1}z)\, \mathrm{d}\omega\ .
	\end{align*}
	However, the dependence of the scattering rate $b$ on the velocities 
	doesn't seem to be treatable with the above methods.  
\end{remark}

The proof of Theorem \ref{thm: boltzmann-main} essentially 
deviates from the one-dimensional case in only two places: the sum rule and 
the discrete approximation of the integrals. We begin by proving an 
analogue of Theorem \ref{thm: Sumrule}. Most of the steps for the computation of 
the matrix $Z$ in \eqref{average} are the same. What remains is to compute
\begin{align*}
	Z :=  \sum_{\alpha_1, \dots, \alpha_k} \lambda_{\alpha_1} \cdots \lambda_{\alpha_k} \int_{\mathbb{S}^2\times \cdots \times \mathbb{S}^2} \mathrm{d}\omega_1\cdots \mathrm{d}\omega_k \, 
	A_k(\underline{\alpha},\underline{\omega}) A_k(\underline{\alpha},\underline{\omega})^T \ , 
\end{align*}
which is somewhat different for the case of Boltzmann-Kac collisions. Recall 
that $A_k(\underline{\alpha},\underline{\omega})$ is the upper left $3M\times 
3M$ block of $\smash{[\prod_{j=1}^k r_{\alpha_j}(\omega_j)]^{-1}}$, i.e., 
 \begin{align*}
 	A_k(\underline{\alpha},\underline{\omega}) = P_{3M} [\Pi_{j=1}^k 
r_{\alpha_j}(\omega_j)]^{-1} P_{3M}^T \ 
 \end{align*}
with the projection $P_{3M} = \begin{pmatrix}
	I_{3M} & 0
\end{pmatrix}$ from $\R^{3M+3N}\to \R^{3M}$. In 
particular, by linearity, 
\begin{align*}
	Z = P_{3M} \left(\sum_{\alpha_1,\dots,\alpha_k} 
\lambda_{\alpha_1}\cdots\lambda_{\alpha_k} \int_{(\mathbb{S}^2)^k} 
\mathrm{d}\underline{\omega} \, 
	\left[\prod_{j=1}^k r_{\alpha_j}(\omega_j)\right]^{-1} \begin{pmatrix}
		I_{3M} & 0 \\ 0 & 0
	\end{pmatrix} \left[\prod_{j=1}^k r_{\alpha_j}(\omega_j)\right] \right) 
P _{3M}^T\ .
\end{align*}

As in the proof of Theorem \ref{thm: Sumrule} we have
\begin{lemma}\label{lem:BoltzmannAAT}
Let $\alpha, \beta \geq 0$. Then 
	\begin{align*}
		\sum_{1\leq i < j \leq M+N} \lambda_{ij} \int_{\mathbb{S}^2} \mathrm{d}\omega \, 
	r_{ij}(\omega)^{-1} \begin{pmatrix}
		\alpha I_{3M} & 0 \\ 0 & \beta I_{3N}
	\end{pmatrix} r_{ij}(\omega) = \begin{pmatrix}
		\alpha' I_{3M} & 0 \\ 0 & \beta' I_{3N}
	\end{pmatrix} \ ,
	\end{align*}
	where $\alpha', \beta'$ are related to $\alpha,\beta$ by
	\begin{align*}
		\begin{pmatrix}
			\alpha' \\ \beta'
		\end{pmatrix} = \mathcal P \begin{pmatrix}
			\alpha \\ \beta 
		\end{pmatrix}, \quad \mathcal P = I_2 - \frac{\mu}{3\Lambda} \begin{pmatrix}
			 1 & -1 \\
			- \frac{M}{N} & \frac{M}{N}
		\end{pmatrix}\ .
	\end{align*}
\end{lemma}

Notice that the matrix $\mathcal P$ of Lemma \ref{lem:BoltzmannAAT} has eigenvalues $1$ and $1- \mu/(3\Lambda) \left(1 + M/N\right)$ with corresponding eigenvectors \smash{$\begin{pmatrix} 1& 1\end{pmatrix}^T$} and \smash{$\begin{pmatrix}-N/M & 1 \end{pmatrix}^T$}. Repeated application of Lemma \ref{lem:BoltzmannAAT} then implies, see also the argument in the one-dimensional case, 
\begin{align*}
	& \sum_{\alpha_1, \dots, \alpha_k} \hskip-1ex\lambda_{\alpha_1} \cdots \lambda_{\alpha_k}   \int_{(\mathbb{S}^2)^k} \hskip-1ex\mathrm{d}\underline{\omega} \, 
	\left[\prod_{j=1}^k r_{\alpha_j}(\omega_j)\right]^{-1} \begin{pmatrix}
		\alpha I_{3M} & 0 \\ 0 & \beta I_{3N}
	\end{pmatrix}  \left[\prod_{j=1}^k r_{\alpha_j}(\omega_j)\right]
	= \begin{pmatrix}
		\alpha^{(k)} I_{3M} & 0 \\ 0 & \beta^{(k)} I_{3N}
	\end{pmatrix},
\end{align*}
where
\begin{align*}
	\begin{pmatrix}
		\alpha^{(k)} \\ \beta^{(k)}
	\end{pmatrix} 
	= \mathcal P^k \begin{pmatrix}
		\alpha \\ \beta
	\end{pmatrix}\ .
\end{align*}
Before we prove Lemma \ref{lem:BoltzmannAAT}, let us make an easy observation. 
\begin{corollary}
In the particular case $\alpha=1$, $\beta=0$, we get
\begin{align*}
	Z = \left[ \frac{M}{M+N} + \frac{N}{M+N} \left( 1- \frac{\mu}{3\Lambda} \left(1+ \frac{M}{N}\right)\right)^k \right] I_{3M} \ .
\end{align*}
\end{corollary}

\begin{proof}[Proof of Lemma \ref{lem:BoltzmannAAT}]
	For $1\leq i < j \leq M$ (respectively for $M+1 \leq i < j \leq M+N$) the operators $r_{ij}(\omega)$ only act non-trivially in the first $3M$ (last $3N$) variables. Taking into account that $r_{ij}(\omega)^{-1} \, I  \, r_{ij}(\omega) = I$, we obtain
	\begin{align*}
		\frac{\lambda_S}{M-1} \sum_{1\leq i < j \leq M} \int_{\mathbb{S}^2}\mathrm{d}\omega \, r_{ij}(\omega)^{-1} \begin{pmatrix}
		\alpha I_{3M} & 0 \\ 0 & \beta I_{3N}
	\end{pmatrix} r_{ij}(\omega) &= \frac{M \lambda_S}{2} \begin{pmatrix}
		\alpha I_{3M} & 0 \\ 0 & \beta I_{3N}
	\end{pmatrix}\ ,
	\intertext{and} 
		\frac{\lambda_R}{N-1} \sum_{M+1\leq i < j \leq M+N} \int_{\mathbb{S}^2}\mathrm{d}\omega \, r_{ij}(\omega)^{-1} \begin{pmatrix}
		\alpha I_{3M} & 0 \\ 0 & \beta I_{3N}
	\end{pmatrix} r_{ij}(\omega) &= \frac{N \lambda_R}{2} \begin{pmatrix}
		\alpha I_{3M} & 0 \\ 0 & \beta I_{3N}
	\end{pmatrix}\ .
	\end{align*}
It remains to look at the interaction terms $i=1, \dots, M$ and $j=M+1, \dots, M+N$. Notice that 
\begin{align*}
	&r_{ij}(\omega)^{-1} \begin{pmatrix}
		\alpha I_{3M} & 0 \\ 0 & \beta I_{3N}
	\end{pmatrix} r_{ij}(\omega) \\
	& = \begin{pmatrix}
		\alpha I_{3M} & 0 \\ 0 & \beta I_{3N}
	\end{pmatrix}
	+ \left( \begin{array}{c|c}
		\begin{matrix}
			0 &  & \\
			  &  (\beta - \alpha) \omega \omega^T &  \\
		  & & 0
		\end{matrix}
		   & 0 \\ \hline 
		   0 & 		\begin{matrix}
			0 &  & \\
			  &  (\beta - \alpha) \omega \omega^T &  \\
		  & & 0
		\end{matrix}
	\end{array}\right) \ ,
\end{align*}
	where the non-zero entries in the second summand on the right-hand side 
correspond to the $i^{\text{th}}$, respectively $j^{\text{th}}$, $3\times 3$ 
block on the diagonal.
Since $\int_{\mathbb{S}^2} \mathrm{d}\omega \, \omega\omega^T = 1/3 \,I_3$, we obtain  
	\begin{align*}
		& \frac{\mu}{N} \sum_{i=1}^M \sum_{j=M+1}^{M+N} \int_{\mathbb{S}^2}\mathrm{d}\omega \, r_{ij}(\omega)^{-1} \begin{pmatrix}
		\alpha I_{3M} & 0 \\ 0 & \beta I_{3N}
	\end{pmatrix} r_{ij}(\omega) \\
	& = \mu M \begin{pmatrix}
		\alpha I_{3M} & 0 \\ 0 & \beta I_{3N}
	\end{pmatrix} 
	+ \frac{\mu}{3} (\alpha-\beta) 
	\begin{pmatrix}
		 -I_{3M} & 0 \\ 0 & \frac{M}{N} I_{3N}
	\end{pmatrix}\ .
	\end{align*}
Recall the definition of  $\Lambda = M\lambda_S/2 + N \lambda_R/2 + \mu M$. 
Hence  summation of all the three contributions yields the statement of the 
Lemma. 
\end{proof}

As in the one-dimensional case, in order to apply the geometric Brascamp-Lieb 
inequality Theorem \ref{thm: correlation}, we need to approximate the uniform probability measure 
$\mathrm{d}\omega$ on the sphere by a suitable sequence of discrete measures
as in the one-dimensional case (see Lemma \ref{lem: approximation}). 
Additionally, in each step of the discretization, the constraint 
$\int_{\mathbb{S}^2} \mathrm{d}\omega \, \omega\omega^T = 1/3 I$, has to hold. 
This is important, because it guarantees that the geometric Brascamp-Lieb 
condition, i.e., the sum rule \eqref{eq: sumrule}, holds in each step.

In order to find such an approximation, we parametrize the sphere in the usual 
way by spherical coordinates  
\begin{align*}
	\omega = \omega(\theta, \varphi) = \begin{pmatrix}
		\sin\theta \cos\varphi \\ \sin\theta \sin\varphi \\ \cos\theta 
	\end{pmatrix}
\end{align*}
for $\theta\in [0, \pi]$ and $\varphi \in [0, 2\pi]$. For $K,L\in\N$ we 
introduce the measures  
\begin{align*}
	\Phi_K &:= \frac{\pi}{K} \sum_{j=0}^{2K-1} \delta_{\frac{\pi}{K}j} \quad \text{on } [0, 2\pi], \quad \text{and}\\
	\Theta_L &:= \sum_{i=1}^L \frac{2}{(1-u_i^2)^{3/2} (P_L'(u_i))^2} \delta_{\arccos u_i} \quad \text{on } [0,\pi],
\end{align*}
where $P_L$ is the Legendre polynomial of order $L$ on $[-1,1]$, and $u_i$, $i=1, \dots, L$, are its zeros. 
Then, if $f \in \mathcal{C}[0, 2\pi]$ and $g\in\mathcal{C}[-1,1]$,
\begin{align*}
	\int_{0}^{2\pi} f(\varphi)\,\Phi_k(\mathrm{d}\varphi) = \frac{\pi}{K} \sum_{j=0}^{2K-1} f\left(\frac{\pi}{K}j\right) \to \int_{0}^{2\pi} f(\varphi)\,\mathrm{d}\varphi 
\end{align*}
as $K\to\infty$ as Riemann sum. Furthermore, 
\begin{align*}
	\int_{0}^\pi g(\cos\theta)\,\sin\theta\,\Theta_L(\mathrm{d}\theta) = \sum_{i=1}^L \frac{2}{(1-u_i^2) (P_L'(u_i))^2} g(u_i)
	\to \int_{-1}^1 g(u)\,\mathrm{d}u = \int_0^{\pi} g(\cos\theta)\,\sin\theta\,\mathrm{d}\theta
\end{align*}
as $L\to\infty$ by Gauss-Legendre quadrature . The latter approximation is exact 
for polynomials of order less or equal to $2L-1$. In particular, we have 
\begin{align*}
	\int_{0}^{\pi} \cos^2\theta \,\sin\theta\,\Theta_L(\mathrm{d}\theta) = \int_{0}^{\pi} \cos^2\theta \,\sin\theta\,\mathrm{d}\theta = \frac{2}{3} \ , \  \text{and }
	\int_{0}^{\pi} \sin^3\theta\,\Theta_L(\mathrm{d}\theta) = \int_{0}^{\pi} \sin^3\theta\,\mathrm{d}\theta = \frac{4}{3},
\end{align*}
for all $L\geq 2$. 
It is easy to check that 
\begin{align*}
	\int_{0}^{2\pi} \sin\varphi \cos\varphi \,\Phi_k(\mathrm{d}\varphi)
	&= 0, \\
	\int_{0}^{2\pi} \sin\varphi \,\Phi_k(\mathrm{d}\varphi) = 
	\int_{0}^{2\pi} \cos\varphi \Phi_k(\mathrm{d}\varphi) &= 0,\\
\int_{0}^{2\pi} \sin^2\varphi \Phi_k(\mathrm{d}\varphi) = 
	\int_{0}^{2\pi} \cos^2\varphi \,\Phi_k(\mathrm{d}\varphi) &= \pi,
\end{align*}
for all $K\geq 2$. Consequently,
\begin{align*}
	&\frac{1}{4\pi} \int_{0}^{2\pi} \omega(\theta, \varphi) \omega(\theta, \varphi)^T \,\Theta_L(\mathrm{d}\theta) \Phi_k(\mathrm{d}\varphi) \\
	&= \frac{1}{2K} \sum_{j=0}^{2K-1} \sum_{i=0}^L \frac{\omega\left(\arccos u_i, \pi j/K\right)\omega\left(\arccos u_i, \pi j/K\right)^T}{(1-u_i^2) (P_L'(u_i))^2}  
	= \frac{1}{3} I_3
\end{align*}
for all $K,L\geq 2$. It follows that  $Z$ is not changed by replacing the uniform measure on $\mathbb{S}^2$ by the above discrete approximation, in particular, $Z$ is still proportional to the identity matrix, which guarantees the applicability of the geometric Brascamp-Lieb inequality.

This concludes the proof of Theorem \ref{thm: boltzmann-main}. \hfill $\square$ 

\appendix

\vskip .3 true in
\section{\bf The Geometric Brascamp-Lieb inequality and the entropy inequality}  \label{sec: bl}

In this section we prove Theorem \ref{thm: correlation}. We use the same strategy as in \cite{CarlenLiebLoss} and
\cite{BennettCarberyChristTao} which consists of transporting the functions $f_i$ with the heat kernel in such a way
that the right-hand side of \eqref{eq: brascamplieb} remains fixed while the 
left-hand side of that inequality increases. The results in
\cite{BennettCarberyChristTao} are quite general but for the special case in which the sum rule \eqref{eq: sumrule} 
holds, the proof is quite simple and this is one of the reasons why we include it here.

\begin{proof} [Proof of Theorem \ref{thm: correlation}]  The inequality 
\eqref{eq: brascamplieb} is equivalent to
\begin{equation} \label{brascamplieboriginal}
\int_{\R^M} \prod_{i=1}^K f_i^{c_i} (B_i \mathbf v) \, \mathrm{d} \mathbf v \le 
\prod_{i=1}^K \left( \int_{H_i} f_i(u) \, \mathrm{d}  u 
\right)^{c_i}  \ . 
\end{equation}
This follows from the identity 
\begin{equation*}
 \prod_{i=1}^K \left(e^{-\pi |B_i \mathbf v|^2}\right)^{c_i}=e^{-\pi 
\sum_{i=1}^K (\mathbf v\  c_iB_i^TB_i\ \mathbf v)}  = e^{-\pi |\mathbf v|^2}\ .
\end{equation*}
We transport the functions $f_i$ by the heat flow, that is we define
\begin{equation} \label{heatflow}
f_i(B_i \mathbf v, t) := \frac{1}{(4 \pi t)^{M/2}} \int_{\R^M} e^{- |\mathbf v - 
\mathbf w|^2/(4t)} f_i(B_i \mathbf w) \mathrm{d} \mathbf w \ .
\end{equation}
For the above definition to make sense, we have to show that the right-hand 
side is a function of $B_i \mathbf v$ alone. The condition $B_iB_i^T= I_{H_i}$ 
means that the matrix $P_i=B_i^TB_i$ is an orthogonal projection onto a 
$d_i$ dimensional subspace of $\R^M$. Moreover, $B_iP_i = I_{H_i} B_i = B_i$. 
We rewrite the integral \eqref{heatflow} by splitting it in an integral over 
$\mathbf w'\in {\rm Ran}\,P_i$ and one over integration over $\mathbf 
w''\in {\rm Ran}\,P_i^\perp$. Carrying out the integration over $\mathbf 
w''$ we obtain 
\begin{align*}
f_i(B_i \mathbf v, t) & = \frac{1}{(4 \pi t)^{M/2}} \int_{{\rm Ran }\,P_i} 
\int_{{\rm Ran}\,P_i^\perp} e^{- |(P_i \mathbf v - P_i\mathbf w') |^2/(4t)} 
e^{- 
|(P_i^\perp \mathbf v - \mathbf w'') |^2/(4t)}f_i(B_i P_i \mathbf w) \, 
\mathrm{d} \mathbf w' \mathrm{d}\mathbf w'' \\
& = \frac{1}{(4 \pi t)^{d_i/2}} \int_{{\rm Ran }\,P_i} e^{- |(P_i \mathbf v - 
P_i 
\mathbf w') |^2/(4t)} f_i(B_i P_i \mathbf w') \, \mathrm{d} \mathbf w' \\ 
& = \frac{1}{(4 \pi t)^{d_i/2}} \int_{{\rm Ran }\,P_i} e^{- |(B_i \mathbf v - 
B_i 
\mathbf w') |^2/(4t)} f_i(B_i \mathbf w')\, \mathrm{d} \mathbf w' \\
& = \frac{1}{(4 \pi t)^{d_i/2}} \int_{H_i} e^{- |(B_i \mathbf v - u) 
|^2/(4t)} f_i( u) \, \mathrm{d} u \ .
\end{align*}
where, in the last equality, we have used that $B_i$ maps the range of $P_i$
isometrically onto $H_i$. This justifies \eqref{heatflow}. Moreover, the above 
computation also shows that 
\begin{equation*}
 \int _{H_i} f_i( u, t)\mathrm{d} u=\int _{H_i} f_i(
u)\mathrm{d}  u
\end{equation*}
so that the right-hand side of  the inequality \eqref{brascamplieboriginal} 
does not change under the heat flow.  

We now show that the left-hand side of \eqref{brascamplieboriginal} is 
an increasing function of $t$. It is convenient to set $\phi_i(u, 
t)=\log f_i(u, t)$. Differentiating the function $\phi_i(B_i \mathbf v, 
t)$ with respect to $t$ yields
\begin{align*}
\frac{\mathrm{d}}{\mathrm{d}t} \phi_i(B_i \mathbf v, t) = \Delta_{\mathbf v} \phi_i(B_i \mathbf v, t) + |\nabla_{\mathbf v} \phi_i(B_i \mathbf v, t)|^2 \ .
\end{align*}
Moreover,
\begin{align*}
\frac{\mathrm{d}}{\mathrm{d}t} \int_{\R^M} \prod_{i=1}^K f_i^{c_i} (B_i \mathbf v,t) \, \mathrm{d} \mathbf v  
= \sum_{m=1}^K c_m \int_{\R^M} [\Delta_{\mathbf v} \phi_m(B_m \mathbf v, t) + |\nabla_{\mathbf v} \phi_m(B_m\mathbf v, t)|^2] \prod_{i=1}^K f_i^{c_i} (B_i \mathbf v,t) \, \mathrm{d} \mathbf v  \ . 
\end{align*}
Integrating by parts the term containing the Laplacian yields
\begin{align*} \label{derivative}
& \frac{\mathrm{d}}{\mathrm{d}t} \int_{\R^M} \prod_{i=1}^K f_i^{c_i} (B_i 
\mathbf v,t) \, \mathrm{d} \mathbf v = \\
&\qquad \sum_{m=1}^K c_m  \int_{\R^M} |\nabla_{\mathbf v} \phi_m(B_m\mathbf v, 
t)|^2  \prod_{i=1}^K f_i^{c_i} (B_i \mathbf v,t) \mathrm{d} \mathbf v   \\
&\qquad\qquad -\sum_{m,\ell=1}^K c_m c_\ell   \int_{\R^M} \nabla_{\mathbf v} 
\phi_m(B_m \mathbf v, t) \cdot \nabla_{\mathbf v} \phi_\ell(B_\ell \mathbf v, t) 
 \prod_{i=1}^K f_i^{c_i} (B_i \mathbf v,t)  \mathrm{d} \mathbf v   \ .
\end{align*}
Finally, using that
\begin{align*}
\nabla_{\mathbf v} \phi_m(B_m \mathbf v, t) = B_i^T (\nabla \phi_m)(B_m \mathbf v)
\end{align*}
we get
\begin{align*}
& \frac{\mathrm{d}}{\mathrm{d}t} \int_{\R^M} \prod_{i=1}^K f_i^{c_i} (B_i 
\mathbf v,t) \, \mathrm{d} \mathbf v = \\
&\qquad \sum_{m=1}^K c_m  \int_{\R^M} |B_m^T( \nabla \phi_m) (B_m\mathbf v, 
t)|^2  \prod_{i=1}^K f_i^{c_i} 
(B_i \mathbf v,t) \mathrm{d} \mathbf v \\
&\qquad\qquad -\sum_{m,\ell=1}^K c_m c_\ell   \int_{\R^M} B_m^T (\nabla  \phi_m) 
(B_m \mathbf v, t) \cdot B_\ell^T (\nabla \phi_\ell)(B_\ell \mathbf v, t) 
\prod_{i=1}^K f_i^{c_i} (B_i \mathbf v,t) \mathrm{d} \mathbf v   \ .
\end{align*}
We claim that this expression is non-negative. The vectors $\nabla \phi_m \in 
H_m$ are arbitrary and hence the problem is reduced to proving that for 
any set of vectors $V_m \in H_m$, $m=1,\ldots, K$, it holds
\begin{align*}
\sum_{m=1}^K c_m  |B_m^T V_m|^2 - \sum_{m,\ell=1}^K c_m c_\ell  B_m^T V_m \cdot B_\ell^T V_\ell \, \ge \, 0 \ .
\end{align*}
Recalling that $B_mB_m^T = I_{H_m}$ and setting $Y = \sum_\ell c_\ell B_\ell^T 
V_\ell$ we conclude that it is enough to show that 
\begin{align*}
|Y|^2 \le \sum_{m=1}^K c_m  | V_m|^2 \ .
\end{align*}
This follows easily, since, by applying Schwarz's inequality, we find that
\begin{align*}
|Y|^2 = \sum_{\ell=1}^K c_\ell Y \cdot B_\ell^T V_\ell = \sum_{\ell=1}^K c_\ell 
B_\ell Y \cdot V_\ell \le \left(\sum_{\ell=1}^K c_\ell |B_\ell Y |^2  
\right)^{1/2}
 \left(\sum_{\ell=1}^K c_\ell  |V_\ell|^2 \right)^{1/2} \ .
\end{align*}
Combining this with \eqref{eq: sumrule}, we learn that
\begin{align*}
|Y|^2 \leq \left( Y\cdot  \sum_{\ell=1}^K c_\ell B_\ell^T B_\ell Y   \right)^{1/2}
\left(\sum_{\ell=1}^K c_\ell  |V_\ell|^2 \right)^{1/2} = |Y| \left(\sum_{\ell=1}^K c_\ell  |V_\ell|^2 \right)^{1/2} \ .
\end{align*}
Thus we have that, when applying \eqref{brascamplieboriginal} to the functions 
$f_i(u,t)$, the left hand side is an increasing function of $t$ while 
the right hand side does not depends on $t$.
It is thus enough to show that the inequality holds for large $t$. Using once 
more the sum-rule \eqref{eq: sumrule}, we see that 

\begin{align*}
	\int_{\R^M} \prod_{i=1}^K \frac{1}{(4 \pi t)^{c_i d_i/2}} 
	\left[\int_{H_i} e^{- \frac{|B_i \mathbf v - u|^2}{4t}}  f_i( u) d u 
	\right]^{c_i}  d \mathbf v  
	&=\\ 
	\frac{1}{(4 \pi)^{M/2}} \int_{\R^M} \prod_{i=1}^K \left[\int_{H_i} 
	e^{- \frac{|B_i \mathbf v - t^{-1/2} u|^2}{4}}  f_i( u) d u 
	\right]^{c_i}  d \mathbf v 
	&\stackrel{t\to\infty}{\longrightarrow} \frac{1}{(4 \pi)^{M/2}} \int_{\R^M} e^{- 
	\frac{|\mathbf v|^2}{4}} \prod_{i=1}^K\left[\int_{H_i}   f_i( u) d u 
	\right]^{c_i}  d \mathbf v \\
	&\qquad = \prod_{i=1}^K \left[\int_{H_i}   f_i( u) d u \right]^{c_i} 
\end{align*}
which proves the first part of  Theorem \ref{thm: correlation}. 

To prove the entropy inequality~\eqref{eq: entropymarginalstwo} we 
follow~\cite{CarlenCordero}. Let  $h$ be a non-negative function whose $L^1$ 
norm is one and whose entropy is finite. An elementary computation then shows 
that
\begin{equation*} \label{dualentropy}
\int_{\R^M} h(\mathbf v) \log h(\mathbf v) \, e^{-\pi |\mathbf v|^2} \mathrm{d} \mathbf v = \sup_{\Phi} \left\{ \int_{\R^M} h(\mathbf v) \Phi(\mathbf v) \, e^{-\pi |\mathbf v|^2} \, \mathrm{d} \mathbf v - \log \int_{\R^M} e^{\Phi(\mathbf v)} e^{-\pi |\mathbf v|^2} \, \mathrm{d} \mathbf v\right\} \ .
\end{equation*}
Now, we set
\begin{align*}
\Phi(\mathbf v) = \sum_{i=1}^K c_i \log f_i(B_i \mathbf v) \ .
\end{align*}
This leads to the lower bound
\begin{align*}
& \int_{\R^M} h(\mathbf v) \log h(\mathbf v) \, e^{-\pi |\mathbf v|^2} \, 
\mathrm{d} \mathbf v \\ & \ge \sum_{i=1}^K c_i   \int_{\R^M} h(\mathbf v) 
 \log f_i(B_i \mathbf v) \, e^{-\pi |\mathbf v|^2} \, \mathrm{d} \mathbf v - 
\log \int_{\R^M} \prod_{i=1}^K f_i(B_i \mathbf v)^{c_i} \, e^{-\pi |\mathbf 
v|^2} \, \mathrm{d} \mathbf v \\ 
& \ge \sum_{i=1}^K c_i   \int_{\R^M} h(\mathbf v)
 \log f_i(B_i \mathbf v) \, e^{-\pi |\mathbf v|^2} \, \mathrm{d} \mathbf v - 
\log \left[\prod_{i=1}^K \left(\int_{H_i} f_i(u) \, e^{-\pi 
|u|^2} \, \mathrm{d} u\right)^{c_i} \right], 
\end{align*}
where the second step is a consequence of the Brascamp-Lieb inequality \eqref{eq: brascamplieb}.
\end{proof}

\vskip .3 true in
\section{\bf Proof of Lemma \ref{lem: approximation}}  \label{sec: approximation}

\begin{proof}
For $K$ any positive integer we convolve $\rho(\theta)$ with the non-negative trigonometric polynomial
\begin{align*}
p_K(\theta) := \frac{1}{2K+1} \left( \sum_{k= - K}^K e^{i k\theta}\right)^2 = \sum_{m=-2K}^{2K} \left(1- \frac{|m|}{2K+1}\right)e^{im \theta} \ ,
\end{align*}
and obtain a probability density $\rho_K(\theta)$. The Fourier coefficients of $\rho_K(\theta)$ are given by
\begin{align*}
\widehat \rho_K(m) = \widehat \rho(m) \left(1- \frac{|m|}{2K+1}\right)
\end{align*}
for $|m| \le 2K$ and are zero otherwise. In particular, 
\begin{align*}
\widehat \rho_K(2) - \widehat \rho_K(-2)  =4i \int_{-\pi}^\pi \rho_K(\theta) \sin \theta \cos \theta \, \mathrm{d} \theta = 0 \ .
\end{align*}
With $\rho_K$ we construct the measure
\begin{align*}
\nu_K( d \theta) = \frac{2\pi}{4K+1} \sum_{\ell = -2K}^{2K} \rho_K\left(\frac{2 \pi \ell}{4K+1}\right) \delta\left(\theta - \frac{2\pi \ell}{4K+1}\right) \, \mathrm{d} \theta \ .
\end{align*}
The measure $\nu_K$ is positive since $\rho_K((2 \pi \ell)/(4K+1)) \ge 0$.
Moreover, for all $m \in \Z$ with $|m| \le 2K$ the Fourier coefficients $\widehat \nu_K(m)$ and 
$\widehat \rho_K(m)$ coincide. In particular, we have
\begin{align*}
\int_{-\pi}^\pi \nu_K(\mathrm{d} \theta) \sin \theta \cos \theta =0 \ .
\end{align*}
To see this, we compute
\begin{align*}
\widehat \nu_K(m) = \frac{1}{2\pi} \int_{-\pi}^\pi \nu_K(\theta) e^{-i m \theta} \, \mathrm{d}\theta =\frac{1}{4K+1}  \sum_{\ell=-2K}^{2K } \rho_K\left(\frac{2 \pi \ell}{4K+1}\right) e^{-2\pi i m \ell/(4K+1)} \ 
\end{align*}
for $|m|\le 2K$. Observe that
\begin{align*}
\rho_K\left(\frac{2 \pi \ell}{4K+1}\right) = \sum_{k=-2K}^{2K}  \widehat \rho_K (k)e^{2 \pi i k \ell/(4K+1)} \ ,
\end{align*}
and, as a consequence,
\begin{align*}
\widehat \nu_K(m) =  \frac{1}{4K+1}  \sum_{\ell=-2K}^{2K}   \sum_{k=-2K}^{2K}  \widehat \rho_K(k)e^{2 \pi i \ell (k-m)/(4K+1)} \ . 
\end{align*}
But
\begin{align*}
\sum_{\ell=-2K}^{2K} e^{2 \pi i \ell (k-m)/(4K+1)} = \begin{cases} 4K+1 &{\rm if} \ k=m \\ 0 & {\rm if} \ k \not= m, \end{cases}
\end{align*}
and hence we conclude that 
\begin{equation} \label{eq: equalfourier}
\widehat \nu_K(m)  = \widehat \rho_K(m)
\end{equation}
for $|m| \le 2K$.
It is easy to see that for any continuous function $f$ on $[-\pi,\pi]$, 
\begin{align*}
\lim_{K \to \infty} \int_{-\pi}^\pi f(\theta) \nu_K(\mathrm{d}\theta) = \int_{-\pi}^\pi f(\theta) \rho(\theta) \, \mathrm{d} \theta \ .
\end{align*}

\end{proof}

\providecommand{\mr}[1]{\href{http://www.ams.org/mathscinet-getitem?mr=#1}{MR~#1}}
\providecommand{\zbl}[1]{\href{https://zbmath.org/?q=an:#1}{Zbl~#1}}
\providecommand{\arxiv}[1]{\href{https://arxiv.org/abs/#1}{arXiv~#1}}
\providecommand{\doi}[2]{\href{https://doi.org/#1}{#2}}

\bigskip

\end{document}